\documentclass[a4paper,USenglish,cleveref,autoref,thm-restate]{lipics-v2021}

\pdfoutput=1 
\hideLIPIcs  

\usepackage{listings}
\usepackage{xcolor}
\usepackage{microtype}

\newcommand{\mypar}[1]{\noindent\textbf{#1}}
\newcommand{\myskip}{\vspace{2pt}}

\definecolor{Azure1}{rgb}{0.95,1,1}

\lstdefinestyle{mystyle}{
    backgroundcolor=\color{Azure1},
    columns=fixed,
    basicstyle=\ttfamily\small,
    basewidth=0.5em,
    breakatwhitespace=false,
    breaklines=true,
    captionpos=b,
    frame=single,
    keepspaces=true,
    numbers=left,
    numberstyle=\small,
    numbersep=5pt,
    showspaces=false,
    showstringspaces=true,
    showtabs=false,
    tabsize=2
}

\lstdefinestyle{mystyle2}{
    backgroundcolor=\color{Azure1},
    columns=fixed,
    basicstyle=\ttfamily\tiny,
    basewidth=0.55em,
    breakatwhitespace=false,
    breaklines=true,
    captionpos=b,
    frame=single,
    keepspaces=true,
    numbers=left,
    numberstyle=\tiny,
    numbersep=5pt,
    showspaces=false,
    showstringspaces=true,
    showtabs=false,
    tabsize=2
}

\crefname{lstlisting}{Listing}{Listings}

\newtheorem{assumption}{Assumption}
\crefname{assumption}{Assumption}{Assumptions}

\title{Universal Finite-State and Self-Stabilizing Computation in Anonymous Dynamic Networks} 

\titlerunning{Universal Finite-State and Self-Stabilizing Computation in ADNs} 

\author{Giuseppe A. Di Luna\footnote{Both authors contributed equally to this research.}}{DIAG, Sapienza University of Rome, Italy}{diluna@diag.uniroma1.it}{}{}
\author{Giovanni Viglietta\footnotemark[\value{footnote}]} {Department of Computer Science and Engineering, University of Aizu, Japan}{viglietta@gmail.com}{}{}

\authorrunning{Giuseppe A. Di Luna and Giovanni Viglietta}

\Copyright{Giuseppe A. Di Luna and Giovanni Viglietta} 



\category{} 
\relatedversion{} 

\acknowledgements{}

\nolinenumbers 

\begin{document}

\maketitle

\begin{abstract}
A communication network is said to be \emph{anonymous} if its agents are indistinguishable from each other; it is \emph{dynamic} if its communication links may appear or disappear unpredictably over time. Assuming that each of the $n$ agents of an anonymous dynamic network is initially given an input, it takes $2\tau n$ communication rounds for the agents to compute an arbitrary (frequency-based) function of such inputs (Di Luna--Viglietta, DISC~2023), where $\tau$ is a parameter called \emph{dynamic disconnectivity}, and measures how far the network is from being always connected (for always connected dynamic networks, $\tau=1$).

It is known that, without making additional assumptions on the network and without knowing the number of agents $n$, it is impossible to compute most functions and explicitly \emph{terminate}. In fact, current state-of-the-art algorithms only achieve \emph{stabilization}, i.e., allow each agent to return an output after every communication round. Outputs can be changed, and are guaranteed to be all correct after $2\tau n$ rounds. Such algorithms rely on the incremental construction of a data structure called a \emph{history tree}, which is augmented at every round. Thus, they end up consuming an unlimited amount of memory, and are also prone to errors in case of memory loss or corruption.

In this paper, we provide a general \emph{self-stabilizing} algorithm for anonymous dynamic networks that stabilizes in $\max\{4\tau n-2\mu,2\mu\}$ rounds (where $\mu$ measures the smallest amount of potentially corrupted history initially stored by any agent), as well as a general \emph{finite-state} algorithm that stabilizes in $\tau(2n^2+n)$ rounds.

Our work improves upon previously known methods that only apply to eventually static networks (Boldi--Vigna, Dist.\ Comp.~2002). In addition, we develop new fundamental techniques and operations involving history trees, which are of independent interest.
\end{abstract}

\section{Introduction}\label{s:1}
\mypar{Dynamic networks.} Technologies such as wireless sensor networks, software-defined networks, and ad-hoc networks of smart devices have made network topologies with frequent and ongoing changes increasingly common. In the algorithmic study of such \emph{dynamic networks}, a standard model assumes a set of agents that communicate synchronously by broadcast, where the topology is always connected but the communication links may change unpredictably from round to round.

\myskip
\mypar{Anonymous networks.} A classic assumption is that each agent has a unique ID, and efficient algorithms are known for solving many problems in this context~\cite{CFMS15,KLO11,KLO10,KOM11,MS18,DW05}. An alternative model is where agents are \emph{anonymous}, and are therefore initially indistinguishable. The study of anonymous networks is both practically and theoretically significant, with prolific research spanning several decades~\cite{BV01,CDS06,CGM08,JMM12,FPP00,SUW15,YK88}.

\myskip
\mypar{Computing in anonymous dynamic networks.} Recent work introduced the \emph{history tree} as a tool for studying anonymous dynamic networks that are connected at every round and have a unique leader~\cite{DV22}. This structure can be constructed via a distributed algorithm using the local information of agents; unlike similar structures, it has a temporal dimension that completely captures a network's dynamism. The theory of history trees was later extended to networks with no leader and networks with multiple leaders, as well as networks that are not necessarily connected at every round~\cite{DVdisc} and congested networks~\cite{DVcong}.

The introduction of history trees enabled the design of optimal linear-time \emph{universal algorithms} for all the aforementioned network scenarios (in a sense that will be elucidated in \cref{s:2}). However, in the leaderless setting considered in this paper, universal computation with \emph{termination} is impossible without additional knowledge about the network (such as the number of its agents). Without such knowledge, universal algorithms for leaderless networks can only \emph{stabilize} on the correct output without explicitly terminating.

\myskip
\mypar{Self-stabilization.} In real systems, temporary faults such as memory corruption or message losses are unavoidable; algorithms that withstand such faults are called \emph{self-stabilizing}. In \emph{static} anonymous networks, universal self-stabilizing algorithms exist~\cite{BV02b}. However, to the best of our knowledge, a similar result for dynamic networks has never appeared in the literature. It is therefore desirable to adapt the techniques in~\cite{DV22,DVdisc} to achieve self-stabilization.\footnote{Note that terminating algorithms cannot be self-stabilizing, because an adversary could corrupt the initial memory states of all agents, causing them to erroneously terminate with an arbitrary output.}

\myskip
\mypar{Finite-state computation.} We remark that the stabilizing algorithms from previous work on anonymous dynamic networks use unlimited local memory.\footnote{The terminating algorithms discussed in~\cite{DV22,DVdisc} have bounded-memory implementations, and so they are automatically finite-state.} For instance, if history trees are used, all agents are required to continually expand this data structure in their internal memory. It would be beneficial to develop a version of the stabilizing algorithms in~\cite{DV22,DVdisc} that uses finite memory, polynomial in the network size and relevant connectivity parameters, which is crucial for implementation in real systems with limited memory.

\subsection*{Contributions and technical challenges}
In this paper we present two main results that address the weaknesses discussed above. 
Preliminary versions of these results were announced at SIROCCO~2024~\cite{sirocco24} and appeared in the proceedings of OPODIS~2024~\cite{DVopodis}. 
The present version includes several improvements: an expanded section on related work, complete proofs for all statements, revised and reorganized text throughout, pseudocode for all algorithms, extensions of all results to disconnected networks, and a better running time for the algorithm in \cref{s:5}. All running times are expressed in terms of the number of agents in the network $n$, as well as the \emph{dynamic disconnectivity} $\tau$, which is a parameter akin to the dynamic diameter. This parameter is defined in \cref{s:2} and was introduced in \cite{DVdisc,sirocco24}.

The first result, presented in \cref{s:4}, is a self-stabilizing version of the stabilizing algorithm proposed in~\cite{DVdisc}. While the original algorithm stabilizes in $2\tau n$ rounds in a network of $n$ agents, our self-stabilizing version takes $\max\{4\tau n-2\mu,2\mu\}$ rounds (where $\mu$ measures the smallest amount of potentially corrupted history initially stored by any agent). This algorithm relies on a novel operation on history trees called \emph{chop}, which deletes old information and eventually removes corrupted data from memory. The chop operation is also used to merge history trees of different heights. We believe the dependence on $\mu$ to be unavoidable in dynamic networks.

The second result, presented in \cref{s:5}, is a finite-state adaptation of the stabilizing algorithm from~\cite{DVdisc}, resulting in an algorithm that uses memory polynomial in $n$ and $\tau$ and stabilizes in $\tau (2n^2+n)$ rounds. The idea of this algorithm is to avoid updating an agent's state for a round if it receives no new relevant information. This seemingly simple strategy has to be carefully implemented to guarantee the algorithm's correctness. In fact, not updating an agent's state causes a ``desynchronization'' of the system that has to be dealt with properly. This leads to the formulation of a generalized theory of history trees, which is of independent interest and may have important applications to asynchronous networks.

Neither of the above algorithms requires any prior knowledge of $n$ or $\tau$. However, \cref{s:3} offers an extra algorithm that is both self-stabilizing and finite-state, and stabilizes in $\tau(2n-2)$ rounds, which is optimal. The downside is that it assumes $n$ and $\tau$ to be known.

\myskip
\mypar{Overview of the algorithms.} All our algorithms let each agent maintain a local representation of the network's recent communication history, which agents exchange and merge at every round. Old information is periodically discarded while the recent history needed to compute the output is preserved. When $n$ and $\tau$ are known, agents retain a fixed window of sufficient length; when they are unknown, self-stabilization is achieved by gradually eliminating arbitrary initial information as a correct recent history grows. Our finite-state algorithm instead prevents local histories from growing indefinitely by suspending updates once agents have obtained mutually consistent information sufficient to compute the output, and resuming them whenever a disagreement is detected.

\subsection*{Related Work}
The study of computation in anonymous dynamic networks has primarily focused on two key problems: the \emph{Counting problem}, i.e., determining the number of agents $n$ in a network with a leader~\cite{DB15,DBBC14b,DBCB13,KM18,KM19,KM20,KM21,KM22}, and the \emph{Average Consensus problem}, i.e., computing the weighted average of the agents' inputs in a leaderless network~\cite{BT89,CL18,CL22,C11,KM21,NOOT09,O17,OT11,T84,YSSBG13}. Several lines of research have led to the development of algorithms addressing these issues, utilizing a combination of local averaging or mass-distribution methods with advanced termination strategies.

A completely different technique involving history trees was recently employed to develop optimal universal algorithms that stabilize within $2n$ rounds in networks with a leader~\cite{DV22}. This method has also been successfully adapted to multi-leader, leaderless, and disconnected networks~\cite{DVdisc}. These works provided tight linear-time stabilizing algorithms for all computable functions in anonymous dynamic networks.

To the best of our knowledge, none of the aforementioned solutions is tolerant to memory corruption. The most relevant work regarding self-stabilizing computation in anonymous networks is~\cite{BV02b}, which proposed a universal protocol for self-stabilizing computation limited to static (or eventually static) networks. This protocol stabilizes in $n+d$ rounds, where $d$ is the diameter of the network. However, this technique cannot be used in dynamic networks, as it relies on the network's static nature both to eliminate corrupted information quickly and to define the agents' state-update rule. Technically, in~\cite{BV02b}, an agent receives information from its neighbors and compares it with its own, keeping only the common portion up to the first discrepancy; it is shown that this rule eliminates corrupted information from the network within $d$ rounds. In a highly dynamic network, however, this strategy may cause agents to frequently reset their memory without stabilizing, even in the absence of memory corruption. In fact, we believe that dynamic networks do not allow the elimination of corrupted data in a time independent of the amount of such corrupted data.

Regarding finite-state computation, we focus on stabilizing or converging algorithms, since the terminating algorithms in~\cite{DV22,DVdisc} already have bounded-memory implementations. The current state of the art for stabilization speed achieves linear time but uses unlimited memory~\cite{DVdisc}. A majority of works on leaderless networks have focused on the Average Consensus problem~\cite{BT89,CL18,CL22,C11,KM21,NOOT09,O17,OT11,T84,YSSBG13}; most of these algorithms use real numbers of arbitrary precision, and therefore require unlimited memory, and converge instead of stabilizing. An exception is~\cite{NOOT09}, which proposes a quantized algorithm that, in our network model, stabilizes in $O\left(n^{3} \log (nQ)\right)$ rounds, where $Q$ is the number of levels used in the quantization. Of course, the use of quantization introduces an error in the average calculation that depends on $Q$.\footnote{This error could likely be eliminated by making $Q$ dependent on $n$, but this would require prior knowledge of $n$. However, in this case, the algorithm in~\cite{DVdisc} already terminates in $2n$ rounds using memory polynomial in $n$.}
We remark that our finite-state algorithm in \cref{s:5} not only has a better stabilization time but is also exact and does not require any prior knowledge of the network.

\section{Model Definition and Basic Structures}\label{s:2}
\mypar{Model of computation.} A \emph{dynamic network} is defined as an infinite sequence $\mathcal G=(G_t)_{t\geq 1}$, where $G_t=(V,E_t)$ is an undirected multigraph, $V=\{p_1, p_2, \dots, p_n\}$ is a system of $n$ \emph{anonymous agents},\footnote{The network is leaderless: no agent is supplied with an identifier or designated as a leader.} and $E_t$ is a multiset of edges representing \emph{links} between agents.\footnote{The parallel links in multigraphs can serve as a model for the multi-path propagation of radio waves. Of course, all the results in this paper also apply to networks modeled as simple graphs.} We assume that link multiplicities are polynomially bounded in $n$.\footnote{This assumption is used only to express the memory bounds in terms of $n$ and $\tau$, but our algorithms remain correct for arbitrary finite link multiplicities. If $M$ is a uniform upper bound on these multiplicities, the memory bounds acquire a logarithmic dependence on $M$.}

Each agent $p$ has an immutable input $\lambda(p)$, chosen from a fixed finite alphabet and available throughout the execution. Agents also have internal states, which are initially determined by their inputs (thus, agents with the same input start in the same state). We refer to the initial configuration as \emph{round~$0$}. Then, at every \emph{round~$t\geq 1$}, each agent sends a message to its neighbors in $G_t$ through all its incident links. Every agent creates its message as a function of its current state and sends the same message to all its current neighbors. During round~$t$, each agent reads all messages coming from its neighbors (in no particular order). Based on the information therein, the agent updates its internal state according to a local algorithm, the same for all agents. An agent may also return an \emph{output} at the end of a round. After that, the next round starts.

Note that our network model is \emph{synchronous}, i.e., all agents are assumed to send and receive their messages simultaneously, and a new round starts at the same time for all agents.

\myskip
\mypar{Universal computation.} If the outputs of all agents remain unchanged starting at round~$t$, the system is said to \emph{stabilize} in $t$ rounds. If executing a certain algorithm causes a system to stabilize regardless of the inputs assigned to its $n$ agents, we view such an algorithm as computing a function from an $n$-tuple of inputs to an $n$-tuple of outputs.

Not all functions can be computed by a system of anonymous agents. It is known that, without making extra assumptions on the network, the agents can essentially only compute the \emph{Input Frequency function}, which requires every agent to output the set of pairs $\{(z,f_\lambda(z))\mid f_\lambda(z)>0\}$, where $f_\lambda(z)=|\{p\in V\mid\lambda(p)=z\}|/n$. For example, if five agents have inputs $a,a,a,b,c$, every agent must output $\{(a,3/5),(b,1/5),(c,1/5)\}$~\cite{DVdisc}.\footnote{The Input Frequency function does not require agents to determine $n$. In particular, if every agent has input $z$, the correct output is $\{(z,100\%)\}$ regardless of the topology or network size; indeed, without a leader, rings of different sizes with identical inputs are indistinguishable to their agents, so exact counting is impossible in general.}

Once this fundamental function has been computed, the system can then immediately compute any \emph{frequency-based} function, i.e., any function that depends only on input percentages, such as the weighted average of the inputs, etc.\footnote{If the network has a unique leader or the number of agents is known, computing the Input Frequency function also allows the system to determine the exact number of agents that were assigned each input.} For this reason, any algorithm that allows a system to compute the Input Frequency function for any multiset of inputs assigned to the agents is said to be \emph{universal}.

\begin{figure}
\centering
\includegraphics[scale=0.5]{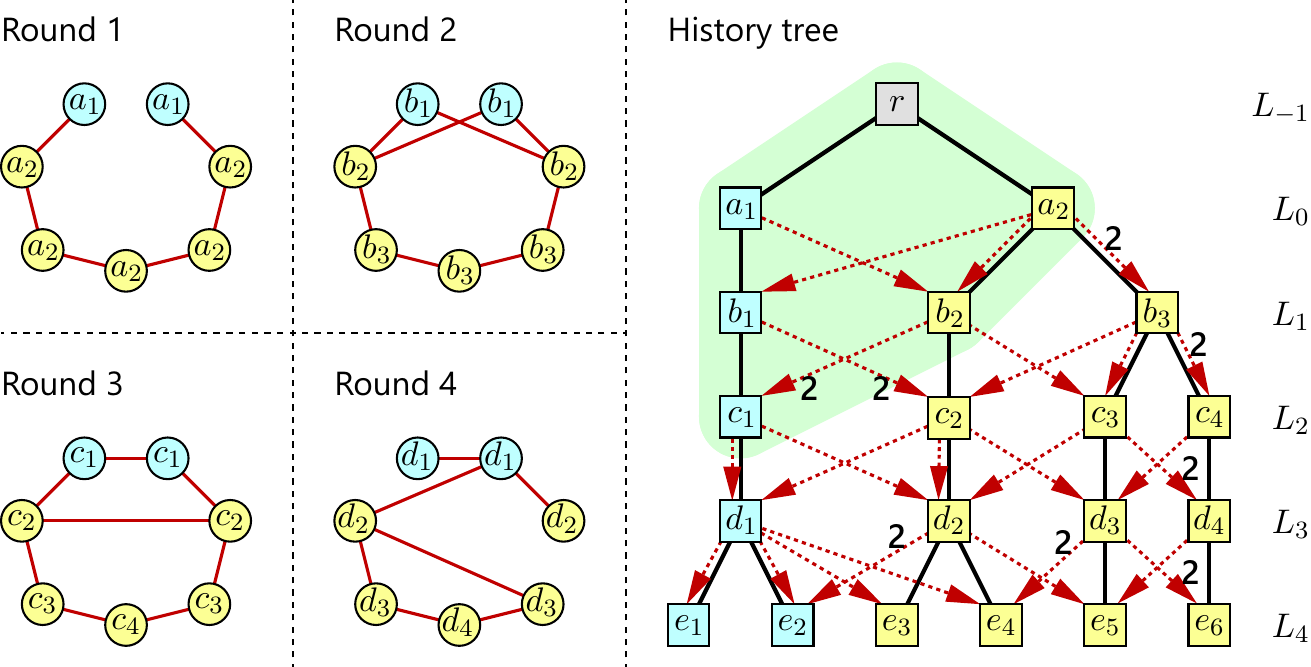}
\caption{The first communication rounds of an anonymous dynamic network of $n=7$ agents and the corresponding levels of its history tree. Initial inputs are represented by agents' colors (cyan or yellow). Labels on agents and nodes are added solely for the reader's convenience to indicate classes of indistinguishable agents; they are not available to the agents. The portion of the history tree with a green background is the vista of the two agents labeled $c_1$ after two communication rounds. Level $L_2$ is a counting level, which yields the equations $|c_1|=|c_2|=|c_3|=2|c_4|$. These equations determine relative anonymities and hence the input frequencies. If $|c_1|=2$ were known from additional information not assumed here, one could also infer that $n=7$.}
\label{fig:ht1}
\end{figure}

\myskip
\mypar{History trees.} \emph{History trees} were introduced in~\cite{DV22} as a tool of investigation for anonymous dynamic networks; an example is in \cref{fig:ht1}. A history tree is a representation of a dynamic network given some inputs to its agents. It is an infinite graph whose nodes are partitioned into \emph{levels} $L_t$, with $t\geq -1$; each node in $L_t$ represents a class of agents that are \emph{indistinguishable} at the end of round~$t$ (with the exception of $L_{-1}$, which contains a single root node $r$ representing all agents). The definition of distinguishability is inductive: at the end of round~$0$, two agents are distinguishable if and only if they have different inputs. At the end of round~$t\geq 1$, two agents are distinguishable if and only if they were already distinguishable at round~$t-1$ or if they have received different multisets of messages during round~$t$. (We refer to ``multisets'' of messages, as opposed to sets, because multiple copies of identical messages may be received; thus, each message has a \emph{multiplicity}.) The \emph{anonymity} of a node $v$ in a history tree is the number of agents represented by $v$, denoted by ${\bf a}(v)$.

Every node (other than the root $r$) has a label indicating the input of the agents it represents. There are also two types of edges connecting nodes in adjacent levels. The \emph{black edges} induce an infinite tree spanning all nodes, rooted at node $r\in L_{-1}$. The presence of a black edge $(v, v')$, with $v\in L_{t}$ and $v'\in L_{t+1}$, indicates that the \emph{child node} $v'$ represents a subset of the agents represented by the \emph{parent node} $v$. The \emph{red multi-edges} represent communications between agents. The presence of a red edge $(v, v')$ with multiplicity $m$, with $v\in L_{t}$ and $v'\in L_{t+1}$, indicates that, at round~$t+1$, each agent represented by $v'$ receives $m$ (identical) messages from agents represented by $v$.

The \emph{vista} $\mathcal V$ of an agent $p$ at round~$t\geq0$ is the subgraph of the history tree that is spanned by all the shortest paths (using black and red edges indifferently) from the root $r$ to the node in $L_t$ representing $p$ (\cref{fig:ht1} shows an example of a vista).\footnote{In previous literature on history trees, vistas were called ``views''~\cite{DV22,DVdisc,sirocco24}. In this paper we changed our terminology to avoid confusion with the loosely related concept of ``views'' introduced by Yamashita and Kameda in the study of static networks~\cite{YK88}.} The node in $\mathcal V$ representing $p$ at round $t$ is the \emph{bottom node} of the vista; this is the node of $\mathcal V$ farthest from $r$.

\begin{figure}
\centering
\includegraphics[scale=0.5]{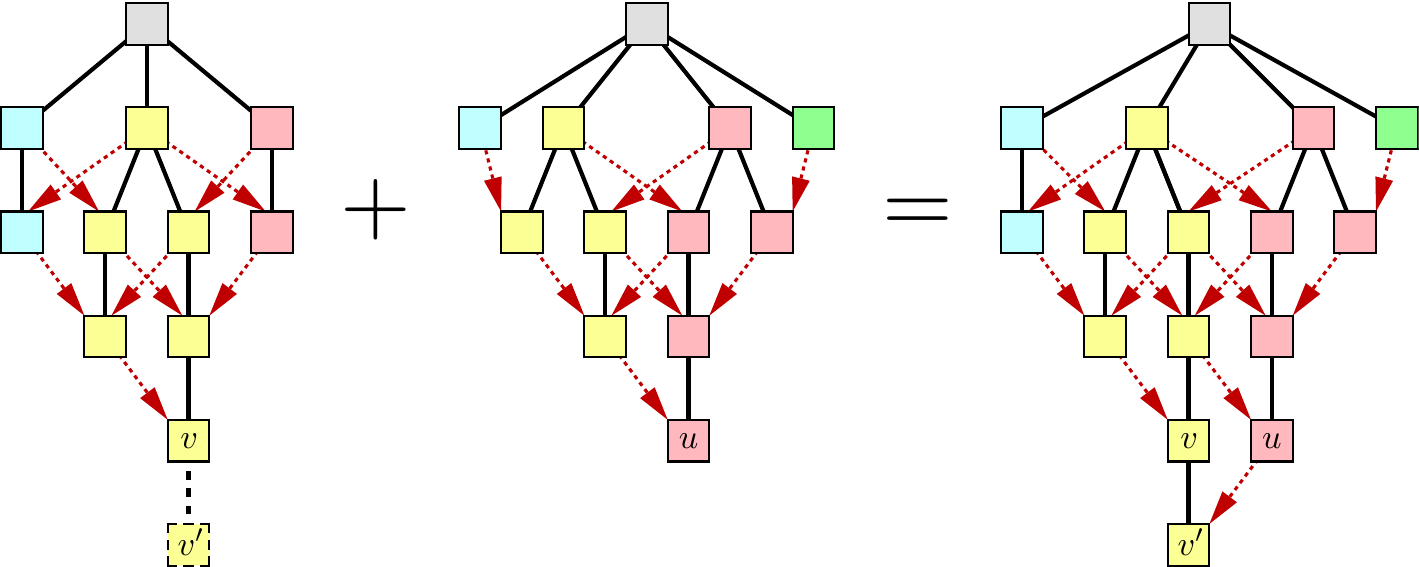}
\caption{Updating the vista of an agent represented by node $v$ after it receives the vista of an agent represented by $u$ as a message. The two vistas are match-and-merged starting from the roots; $v$ also gets a child $v'$, and a new red edge from $u$ to $v'$ is added, representing this interaction.}
\label{fig:merge}
\end{figure}

\myskip
\mypar{Constructing vistas.} It was shown in~\cite{DV22} that each agent can construct its vista of the history tree in real time. To achieve this, each agent is required to send its current vista to all its neighbors at every round, and then match-and-merge it with all the vistas it receives from them. An example of this operation is shown in \cref{fig:merge}. At every round, the bottom node $v$ of the vista also receives a new child $v'$, which becomes the new bottom node. After merging an incoming vista, a red edge is added from the bottom node of this vista to $v'$. The multiplicities of such red edges match the number of isomorphic vistas that are being merged.

\myskip
\mypar{Computing with history trees.} Using techniques introduced in~\cite{DV22}, if the network is connected at every round, then the vista of any agent at round $2n-2$ contains enough information to compute the Input Frequency function. This is done by finding a \emph{counting level}, which is a level in the history tree where every node has exactly one child; level $L_2$ in \cref{fig:ht1} is an example.

For any red edge with multiplicity $m_1>0$ directed from a node $v$ in a counting level $L_t$ to the child of another node $u\in L_t$, there must be a red edge with multiplicity $m_2>0$ from $u$ to the child of $v$. Then, we have $m_1{\bf a}(u)=m_2{\bf a}(v)$. By writing the system of such linear equations for the entire level $L_t$ and solving it up to a common factor, we can compute the Input Frequency function without determining $n$. Note that the first counting level $L_t$ occurs when $t\leq n-2$, and it takes at most $t+n$ rounds for all the nodes in $L_t$ and their children to appear in the vistas of all agents. Thus, this algorithm stabilizes in at most $2n-2$ rounds.

\myskip
\mypar{Self-stabilization.} An algorithm computes a function in a \emph{self-stabilizing} manner if it does so regardless of the initial states of the agents. That is, even if the states of the agents consist of incorrect history trees, a self-stabilizing algorithm is able to eventually erase all errors and stabilize on the correct output. Clearly, the above algorithm is not self-stabilizing.

\myskip
\mypar{Finite-state computation.} An algorithm is \emph{finite-state} if, for every fixed $n$ and $\tau$, the state of each agent after every round can be encoded by at most $f(n,\tau)$ bits on all $\tau$-union-connected networks of size $n$; the parameters need not be known to the agents. In self-stabilizing executions, the arbitrary initial states at round~$0$ are excluded from this bound, which applies after the first state update. Since the above algorithm constructs ever-growing vistas, it is not finite-state.

\myskip
\mypar{Disconnected networks.}
The algorithm described above assumes that the network is connected at every round, that is, the multigraph $G_t$ is connected for all $t \geq 1$. The notion of counting level was later extended to disconnected networks in~\cite{DVdisc}. It is assumed that the network is \emph{$\tau$-union-connected}, meaning that there exists a constant $\tau$, called the \emph{dynamic disconnectivity}, such that the union of any $\tau$ consecutive $G_t$'s is connected (the union of multigraphs is obtained by adding their adjacency matrices).\footnote{The dynamic disconnectivity is related to the usual notion of dynamic diameter $d$~\cite{KLO10} by the tight inequalities $\tau \leq d \leq \tau(n-1)$.} Under this assumption, the Input Frequency function can be computed with stabilization within $\tau (2n-2)$ rounds.

\section{Universal Computation in Networks of Known Size}\label{s:3}
We first assume that the number of agents $n$ in an anonymous dynamic network is known. Using this knowledge, we develop a relatively straightforward algorithm for connected networks that is both self-stabilizing and finite-state and has an optimal stabilization time of $2n-2$ rounds. We then generalize this algorithm to $\tau$-union-connected networks, where the stabilization time becomes $\tau(2n-2)$, provided that both $n$ and $\tau$ are known.

\myskip
\mypar{Chop operation.} The algorithm is based on an operation on history tree vistas called \emph{chop}, illustrated in \cref{fig:self}. Given a vista, eliminate the nodes in level $L_0$ (as well as their incident black and red edges) and connect the root to all nodes in level $L_1$ via black edges. As a result, all non-root nodes are shifted by one level: those that were in $L_1$ are now in $L_0$, etc.

After that, we have to ensure that the resulting vista is well-formed. Note that every node $v$ in a vista $\mathcal V$ is the bottom node of a maximal vista $\mathcal V_v\subseteq \mathcal V$, called the \emph{sub-vista} of $\mathcal V$ at $v$. To complete the chop operation, we repeatedly merge the nodes whose sub-vistas are isomorphic. This is done level by level, starting from the nodes in $L_0$, then the nodes in $L_1$, etc. As a consequence, when two nodes $u,v\in L_t$ are found whose sub-vistas $\mathcal V_u$ and $\mathcal V_v$ are isomorphic, they must be siblings. This is because the parents of $u$ and $v$ have isomorphic sub-vistas as well, and therefore they have already been merged as nodes of $L_{t-1}$.

Merging $u$ and $v$ is done as follows. First we delete one of them, say $u$. The children of $u$ now become children of $v$. The red edges inbound to $u$ are simply eliminated, while the red edges outbound from $u$ are now redirected as outbound from $v$. Specifically, for every $w\in L_{t+1}$, let $m$ and $m'$ be the multiplicities of the red edges $(v,w)$ and $(u,w)$, respectively (taking either multiplicity to be~$0$ if the corresponding edge is absent); after merging $u$ and $v$, the multiplicity of the red edge $(v,w)$ is $m+m'$ (with multiplicity~$0$ meaning that the edge is absent).

This is justified by the fact that the agents represented by $u$ and $v$ have received identical multisets of messages and are therefore indistinguishable; hence the need to merge the two nodes. Thus, all messages received from agents represented by $u$ and $v$ are now considered identical, and the multiplicities of the corresponding red edges must be added together.

The significance of the chop operation is given by the following result. Essentially, chopping a vista is equivalent to ``forgetting'' the first communication round of the network that originated that vista.

\begin{figure}
\centering
\includegraphics[scale=0.5]{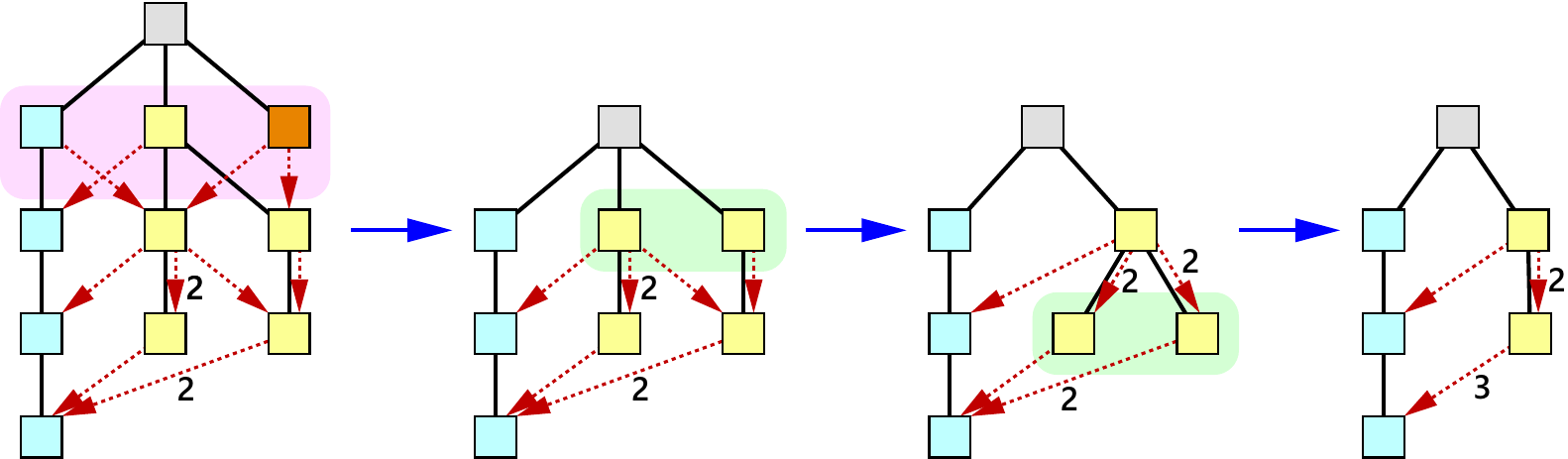}
\caption{The chop operation on a history tree vista: eliminate level $L_0$, and then repeatedly merge nodes whose corresponding sub-vistas are isomorphic, also combining their outgoing red edges.}
\label{fig:self}
\end{figure}

\begin{lemma}\label{l:chop}
Let $\mathcal G=(G_1,G_2, G_3, \dots)$ be a network, and let $\mathcal V$ be the vista of an agent $p$ at round~$t\geq 1$. Then, the vista $\mathcal V'$ obtained by chopping $\mathcal V$ is isomorphic to the vista of $p$ at round~$t-1$ in the network $\mathcal G'=(G_2,G_3, G_4, \dots)$, where agents' inputs are as in $\mathcal G$.
\end{lemma}
\begin{proof}
We define a \emph{compound vista} as the result of the match-and-merge operation on a multiset of vistas of equal height. We will prove that our claim holds not only for the vista of an agent, but more generally for the compound vista of any set of agents in the same network. The concept of sub-vista of a compound vista and the chop operation on a compound vista are defined exactly as for single vistas.

So, our generalized claim is as follows. Let $\mathcal G=(G_1,G_2, G_3, \dots)$ be a network, and let $\widetilde{\mathcal V}$ be the compound vista consisting of the vistas at round~$t\geq 1$ of the agents in a subset $S$ of a system of agents. Then, the compound vista $\widetilde{\mathcal V}'$ obtained by chopping $\widetilde{\mathcal V}$ is isomorphic to the compound vista consisting of the vistas of the agents in $S$ at round~$t-1$ in the network $\mathcal G'=(G_2,G_3, G_4, \dots)$, where agents' inputs are as in $\mathcal G$.

Our proof is by induction on $t$. For $t=1$ the proof is straightforward, as the deepest level of the compound vista $\widetilde{\mathcal V}$ is $L_1$. When $\widetilde{\mathcal V}$ is chopped, its level $L_0$ is removed, all nodes in $L_1$ get linked to the root becoming level $L_0$, and all nodes with the same input are merged (because their sub-vistas are isomorphic). Thus, the resulting compound vista $\widetilde{\mathcal V}'$ accounts for all the agents in $S$ at round~$0$ in $\mathcal G$, i.e., it has one node in level $L_0$ for each distinct input value occurring among the agents in $S$. However, since the agents of $\mathcal G$ and the agents of $\mathcal G'$ have the same inputs, $\widetilde{\mathcal V}'$ also represents $S$ at round~$0$ in $\mathcal G'$.

Now assume that our claim holds at round~$t-1\geq 1$, and let us prove that it holds at round~$t$. We define $\widetilde{\mathcal V}''$ as the compound vista representing the agents in $S$ at round~$t-1$ in the network $\mathcal G'$. Our goal is to prove that $\widetilde{\mathcal V}'$ and $\widetilde{\mathcal V}''$ are isomorphic.

Consider any node $v$ in level $L_t$ of $\widetilde{\mathcal V}$, representing a set of agents $S_v\subseteq S$. Let the $k$ red edges inbound to $v$ have multiplicities $m_1$, $m_2$, \dots, $m_k$, and let these red edges come from nodes $u_1$, $u_2$, \dots, $u_k$ in level $L_{t-1}$ of $\widetilde{\mathcal V}$, whose sub-vistas in $\widetilde{\mathcal V}$ are $\mathcal U_1$, $\mathcal U_2$, \dots, $\mathcal U_k$.

The compound vista $\widetilde{\mathcal W}$ formed by match-and-merging $\mathcal U_1$, $\mathcal U_2$, \dots, $\mathcal U_k$ describes the agents in a set $S'$ at round~$t-1$: the agents in $S'$ are those that send messages to agents in $S_v$ at round~$t$ in $\mathcal G$. By the inductive hypothesis, chopping $\widetilde{\mathcal W}$ results in a compound vista $\widetilde{\mathcal W}'$ that describes the agents in $S'$ in the network $\mathcal G'$ at round~$t-2$ (by convention, we assume that at \emph{round~$-1$} all agents are indistinguishable, have no inputs, and are all represented by the root of the history tree). Thus, $\widetilde{\mathcal W}'$ is a subtree of $\widetilde{\mathcal V}''$. However, by definition of the chop operation, chopping $\widetilde{\mathcal W}$ (which is a subtree of $\widetilde{\mathcal V}$) also yields a subtree of $\widetilde{\mathcal V}'$ (which is obtained by chopping $\widetilde{\mathcal V}$). Thus, $\widetilde{\mathcal V}'$ and $\widetilde{\mathcal V}''$ both contain isomorphic copies of $\widetilde{\mathcal W}'$.

Chopping $\widetilde{\mathcal V}$ may cause some nodes among $u_1$, $u_2$, \dots, $u_k$, say $u_{i_1}$, $u_{i_2}$, \dots, $u_{i_j}$, to merge into a single node $u'$. Accordingly, the agents represented by $u_{i_1}$, $u_{i_2}$, \dots, $u_{i_j}$ have isomorphic vistas at round~$t-2$ in $\mathcal G'$. Thus, the agents in $S_v$ receive a total of $m=m_{i_1}+m_{i_2}+ \dots +m_{i_j}$ equal messages from the agents represented by $u_{i_1}$, $u_{i_2}$, \dots, $u_{i_j}$. Hence, $m$ is the multiplicity of the red edge in $\widetilde{\mathcal V}''$ connecting the node representing $u'$ with the node $v'$ representing $S_v$. By definition of the chop operation, $m$ is also the multiplicity of the new red edge $(u',v)$ that is created upon merging $u_{i_1}$, $u_{i_2}$, \dots, $u_{i_j}$ into $u'$.

Since the above is true of any red edge inbound to $v'$, we conclude that the sub-vista of $\widetilde{\mathcal V}'$ at $v$ is isomorphic to the sub-vista of $\widetilde{\mathcal V}''$ at $v'$. In turn, this holds for every node $v$ in level $L_t$ of $\widetilde{\mathcal V}$. Moreover, neither $\widetilde{\mathcal V}'$ nor $\widetilde{\mathcal V}''$ contains multiple nodes whose sub-vistas are isomorphic, by definition of the merge operation and by the chop operation, respectively. It follows that there is in fact a well-defined isomorphism between $\widetilde{\mathcal V}'$ and $\widetilde{\mathcal V}''$.
\end{proof}

In particular, \cref{l:chop} immediately implies that the order of chopping and match-and-merge does not matter: performing match-and-merge on the vistas of a system of agents according to the communication graph $G_t$ and then chopping the resulting vistas yields the same result as first chopping each vista and then applying match-and-merge according to $G_t$.

\begin{observation}\label{o:commute}
The operations of chop and match-and-merge commute.
\end{observation}

\mypar{Algorithm overview.} The pseudocode of our algorithm is found in \cref{l:1} of \cref{a:1}. For simplicity, let us assume that $\tau=1$; the straightforward extension to the general case with $\tau\geq 1$ will be discussed later. The goal of the algorithm is for each agent to construct a coherent vista of the history tree representing the $2n-2$ most recent rounds, which enables the stabilizing algorithm from~\cite{DVdisc} to correctly compute the Input Frequency function based on the first counting level in each vista (cf.~\cref{s:2}).

If the state of an agent $p$ has an encoding that exceeds a fixed polynomial bound in $n$ and $\tau$ (chosen large enough to encode every correct vista of the permitted height), or does not represent a well-formed vista of a history tree, or represents a vista with more than $2n-2$ levels (excluding levels $L_{-1}$ and $L_0$), its state is simply re-initialized to a vista containing a root and a single child labeled as the input of $p$. When $p$ receives the vistas of its neighbors, it determines the vista with smallest height (including its own) and trims all other vistas by repeatedly performing the chop operation on them, until all vistas have the same height. Then, the resulting vistas are match-and-merged into $p$'s vista as usual. After that, if the current vista of $p$ has $2n-1$ levels (excluding levels $L_{-1}$ and $L_0$), the oldest level is removed via a chop operation. Finally, if the resulting vista contains a counting level, this is used to output the frequencies of all inputs; otherwise, $p$ returns a default output consisting of its own input with $100\%$ frequency.

\begin{theorem}\label{t:1}
For every $n\geq 1$, there is a self-stabilizing, finite-state universal algorithm that operates in any connected anonymous dynamic network of known size $n$ and stabilizes in at most $2n-2$ rounds.
\end{theorem}
\begin{proof}
We will prove that the above algorithm satisfies the required conditions. Denote the communication network as $(G_1,G_2,G_3,\dots)$, and let $\ell_t=\min\{t,2n-2\}$. We prove by induction on $t\geq 0$ that, at the end of round~$t$, the last $\ell_t$ levels of the vistas of all agents correctly represent the most recent $\ell_t$ communication rounds of the network. More precisely, chopping each agent's vista after round~$t$ until it has $\ell_t$ levels yields vistas isomorphic to those of the network $(G_{t-\ell_t+1}, G_{t-\ell_t+2}, G_{t-\ell_t+3}, \dots)$ after $\ell_t$ rounds (i.e., the chopped vistas collectively represent the communication graphs $G_{t-\ell_t+1}$, $G_{t-\ell_t+2}$, \dots, $G_{t}$).

The property trivially holds at the end of round~$0$. Assume now that it holds at the end of round~$t$, and let us prove that it holds at the end of round~$t+1$. By the inductive hypothesis, an agent $p$ starts round~$t+1$ with a vista $\mathcal V_t$ such that chopping $\mathcal V_t$ until it has $\ell_t$ levels yields a vista $\mathcal V'_t$ representing the communication graphs $G_{t-\ell_t+1}$, $G_{t-\ell_t+2}$, \dots, $G_{t}$. During round $t+1$, the agent $p$ receives the vistas of its neighbors in $G_{t+1}$. Again, by the inductive hypothesis, the last $\ell_t$ levels of the received vistas correctly represent the most recent $\ell_t$ communication rounds: $G_{t-\ell_t+1}$, $G_{t-\ell_t+2}$, \dots, $G_{t}$.

According to our algorithm, $\mathcal V_t$ and the incoming vistas are chopped until their heights are equal. Then, $\mathcal V_t$ is updated and extended according to the standard match-and-merge routine of \cref{s:2}. Since all of the vistas involved have at least $\ell_t$ levels, the resulting updated vista $\mathcal V_{t+1}$ has at least $\ell_{t}+1$ levels.

Let $\mathcal V'_{t+1}$ be the vista obtained by chopping $\mathcal V_{t+1}$ until it has exactly $\ell_{t}+1$ levels. Due to \cref{l:chop,o:commute}, $\mathcal V'_{t+1}$ is isomorphic to the vista we would obtain if we first chopped $\mathcal V_t$ and the incoming vistas until they have $\ell_t$ levels and then match-and-merged them. We conclude that the last $\ell_t+1$ levels of $\mathcal V_{t+1}$ correctly represent the most recent $\ell_t+1$ communication rounds.

According to our algorithm, $\mathcal V_{t+1}$ is further chopped if its height exceeds $2n-2$ levels. Due to \cref{l:chop}, the last $\ell_{t+1}$ levels of the resulting vista correctly represent the most recent $\min\{\ell_{t}+1,2n-2\}=\ell_{t+1}$ levels, as desired. This completes the induction.

Therefore, after $t\geq 2n-2$ rounds, all agents hold vistas of the history tree of the most recent $2n-2$ communication rounds. It follows that from round $2n-2$ onward, all agents consistently return the correct output.

The algorithm is finite-state because after each round every state consists of a vista with $O(n)$ levels whose encoding satisfies the prescribed size bound. If this vista is correct, its size is $O(n^3\log n)$ bits. An encoding that exceeds this bound is discarded by the validation step.
\end{proof}

The stabilization time of this algorithm is asymptotically optimal, since a lower bound of  $2n-O(1)$ rounds was proved in~\cite{DVdisc}, even for algorithms that are not self-stabilizing.

\myskip
\mypar{Disconnected networks.} It is straightforward to extend the previous algorithm to $\tau$-union-connected networks for any $\tau\geq 1$, provided that $\tau$ (as well as $n$) is known in advance.

\begin{corollary}\label{c:1}
For every $n\geq 1$ and $\tau\geq 1$, there is a self-stabilizing, finite-state universal algorithm that operates in any $\tau$-union-connected anonymous dynamic network of size $n$ and stabilizes in at most $\tau(2n-2)$ rounds.
\end{corollary}
\begin{proof}
The algorithm is the same as the one previously discussed, except that it limits vistas to $\tau(2n - 2)$ levels instead of $2n - 2$. By the same argument used in \cref{t:1}, after $\tau(2n - 2)$ rounds, and from then on, all agents have vistas correctly representing the most recent $\tau(2n - 2)$ communication rounds. At this point, the algorithm in~\cite[Theorem~3.2]{DVdisc} correctly computes the Input Frequency function on all vistas.
\end{proof}

\section{Universal Self-Stabilizing Computation}\label{s:4}
We will now propose a self-stabilizing universal algorithm for disconnected anonymous dynamic networks that operates without any knowledge of the number of agents $n$ or the dynamic disconnectivity $\tau$.

\myskip
\mypar{Garbage coefficient.} Recall that a self-stabilizing algorithm should tolerate any amount of incorrect data initially present in the agents' local memory. If such initial data does not encode the vista of a history tree, it is definitely incorrect and can be discarded immediately. However, even if an initial state does encode a well-formed vista, the information therein may be incorrect and possibly deceiving. We define the \emph{garbage coefficient} of an agent as the height of the vista encoded by its initial state. If its initial state does not encode a well-formed vista, its garbage coefficient is defined to be $0$.

\myskip
\mypar{Algorithm overview.} We will now describe our self-stabilizing algorithm. Its pseudocode is found in \cref{l:2} of \cref{a:1}. The structure of this algorithm is similar to the one in \cref{s:3}; in particular, the same chop operation is used to gradually eliminate all incorrect levels from every vista of the history tree. The main difficulty is that, without knowledge of $n$ or $\tau$, it is impossible to determine whether a vista has the desired number of levels that are necessary to perform meaningful computations.

To overcome this difficulty, our strategy is to chop each agent's vista once every two rounds. As a result, the height of each agent's vista grows by one level every two rounds, and eventually acquires the desired amount of correct levels.

\myskip
\mypar{Algorithm details.} To control the chopping of vistas, the state of each agent is augmented with a binary flag that is toggled between $0$ and $1$ at every round. Every time an agent's flag becomes $1$, its vista is chopped. Accordingly, if the current state of an agent does not encode a well-formed vista augmented with a binary flag, the state is discarded and reset to a vista of only two nodes (a root with a single child containing the agent's input) and a flag set to $1$.

An agent's binary flag is also attached to all messages it sends, alongside its vista of the history tree. Upon receiving messages from its neighbors, the agent reduces all vistas (including its own) to the same height via chop operations, in the same way as the algorithm in \cref{s:3} does. However, this is now extended to all flags, as well. Specifically, we define a function $\mathbf{eval}(\mathcal V, b)$, which takes a vista $\mathcal V$ and a flag $b\in \{0,1\}$, and returns $2h_\mathcal V+b$, where $h_\mathcal V$ is the height of $\mathcal V$. Let the \emph{minimum pair} $(\mathcal V^\ast, b^\ast)$ be the pair that minimizes $\mathbf{eval}$ among the ones received from its neighbors, as well as the agent's own pair. Then, all vistas received by the agent, as well as the agent's own vista, are chopped until they have the same height as $\mathcal V^\ast$. After that, the incoming vistas are match-and-merged into the agent's vista, as usual. In addition, the agent's own flag is set to $b^\ast$. The flag is then toggled, and if it becomes $1$ the agent's vista is chopped once. This protocol ensures proper synchronization among agents.

Finally, if the agent's resulting vista contains a counting level, this level is used to compute and output the frequencies of all inputs (as explained in \cref{s:2,s:3}); otherwise, the agent returns a default output consisting of its own input with $100\%$ frequency.

\begin{theorem}\label{t:2}
There is a self-stabilizing universal algorithm that operates in any $\tau$-union-connected anonymous dynamic network of unknown size $n$ (and unknown $\tau$) and stabilizes in at most $\max\{4\tau n-2\mu,2\mu\}$ rounds, where $\mu$ is the minimum garbage coefficient across all agents.
\end{theorem}
\begin{proof}
Given the properties of the chop operation established in \cref{l:chop,o:commute}, the proof of correctness of this algorithm is straightforward and similar to the one in \cref{t:1}. Namely, at the end of any round $t\geq 0$, the last $\min\{t,h_\mathcal V\}$ levels of each vista $\mathcal V$ correctly represent the most recent communication rounds (where $h_\mathcal V$ is the height of $\mathcal V$). Since our algorithm chops each vista roughly once every two rounds, eventually all vistas end up having at least $2\tau n$ levels, all of which are correct. At this point, the algorithm in \cite[Theorem~3.2]{DVdisc} applied to such vistas stably outputs the frequencies of all inputs. It remains to show that this happens within $\max\{4\tau n-2\mu,2\mu\}$ rounds.

Recall that the garbage coefficient of an agent $p$ is defined as the number of levels of the vista initially encoded by the state of $p$. Every time a chop operation is performed on the vista of $p$, the number of these levels is decreased by one unit. We define the \emph{leftover garbage} of $p$ at round~$t\geq 0$ as the garbage coefficient of $p$ minus the number of times the vista of $p$ has been chopped up to round~$t$ (or $0$ if this number is negative). Thus, the leftover garbage measures how many of the initial levels are still present in the vista of an agent. We define $\mu_t$ as the minimum leftover garbage at round~$t\geq 0$ across all agents. Clearly, $\mu_0=\mu$.

Let $p$ be an agent whose $\mathbf{eval}$ value at round~$0$ is minimal (in particular, $p$ has the smallest garbage coefficient). By the way the algorithm updates agents' states, the $\mathbf{eval}$ value of $p$ increases by exactly one unit in each round. Furthermore, whenever an agent $q$ communicates with $p$, the $\mathbf{eval}$ value of $q$ becomes equal to that of $p$. After that, any agent communicating with $q$ acquires the same $\mathbf{eval}$ value as well, etc. In essence, $p$ is ``broadcasting'' its $\mathbf{eval}$ value across the network (cf.~\cite{DV22,DVdisc,KLO10}).

Define $S_t$ as the set of agents whose $\mathbf{eval}$ value is minimal at round~$t$. By the above reasoning, the set $S_t$ gains at least one new agent in every interval of at most $\tau$ rounds. Indeed, by the definition of $\tau$, it takes at most $\tau$ rounds for some agent in $S_t$ to communicate with an agent outside $S_t$. This process continues until $S_t$ contains all agents, which occurs in fewer than $\tau n$ rounds. Thus, within $\tau n$ rounds, all agents have the same $\mathbf{eval}$ value, and therefore the same leftover garbage. Also, every two rounds the leftover garbage of any agent with smallest $\mathbf{eval}$ is reduced by one unit (unless it is already $0$). In other words, $\mu_t$ is decremented every two rounds until it reaches $0$. In particular, $\mu_{2\mu}=0$.

Let us assume that $\mu\leq \tau n$; we will prove that the stabilization time in this case is $\max\{4\tau n-2\mu,2\mu\}=4\tau n-2\mu$. Since $\mu_{2\mu}=0$, after $\max\{2\mu,\tau n\}\leq 4\tau n-2\mu$ rounds all agents have no leftover garbage. We conclude that, after $4\tau n-2\mu=2(2\tau n-\mu)$ rounds, every vista has $\mu+(2\tau n-\mu)=2\tau n$ levels, which are all correct. From this point onward, each agent's vista will consistently represent at least $2\tau n$ previous communication rounds of the network, and hence every agent will return the correct output, due to~\cite[Theorem~3.2]{DVdisc}.

Let us now assume that $\mu>\tau n$; we will prove that the stabilization time in this case is $\max\{4\tau n-2\mu,2\mu\}=2\mu$. By the above argument, after $\max\{2\mu,\tau n\}=2\mu$ rounds, all agents have no leftover garbage. At this time, there are $2\mu>2\tau n$ correct levels in every agent's vista, and so their outputs are all correct, and will be correct in all subsequent rounds.
\end{proof}

Observe that our algorithm's stabilization time exhibits a linear dependence on the minimum garbage coefficient $\mu$. We conjecture that this cannot be avoided in dynamic networks, and so this stabilization time is asymptotically optimal. Incidentally, it is remarkable that the best stabilization time of our algorithm is not achieved when $\mu=0$, but when $\mu=\tau n$.

\section{Universal Finite-State Computation}\label{s:5}
Our final contribution is a finite-state universal algorithm for anonymous dynamic networks that operates without any knowledge of the number of agents $n$ or the dynamic disconnectivity $\tau$. This is achieved by preventing agents from updating their states under certain conditions, which leads to different agents potentially having vistas of different heights. Several new definitions are required.

\myskip
\mypar{Generalized vistas.} Before we state our algorithm, we first have to define a generalized notion of vista. This is necessary because we will have to deal with situations where vistas of different heights must be match-and-merged into each other. The result of such an operation is illustrated in \cref{fig:ssynch}. As usual, when an agent receives some vistas from its neighbors, it attaches a child $v'$ to the bottom node of its own vista, and then match-and-merges the incoming vistas, connecting their bottom nodes to $v'$ via red edges. If the incoming vistas have arbitrary heights, the resulting vista may not have only red edges connecting a level to the next, but also red edges spanning any number of levels, going downward or upward.

Formally, a generalized vista can be defined recursively as the result of the above update operation performed on arbitrary generalized vistas, where the base case is the trivial vista consisting of the single root node. Accordingly, we define the bottom node of a generalized vista as its unique sink, i.e., the unique node with no outgoing red or black edges (where black edges are understood as being directed away from the root). By definition, the bottom node of a vista represents the agents that possess this vista at the end of some round.

The notion of \emph{sub-vista} of a generalized vista $\mathcal V$ at a node $v$ is as in \cref{s:3}: it is the maximal vista included in $\mathcal V$ whose bottom node is $v$. As usual, a node $v$ in $\mathcal V$ represents all agents that at the end of some round have a vista isomorphic to the sub-vista of $\mathcal V$ at $v$.

From now on, when no confusion may arise, the term ``generalized'' will often be omitted.

\begin{figure}
\centering
\includegraphics[scale=0.5]{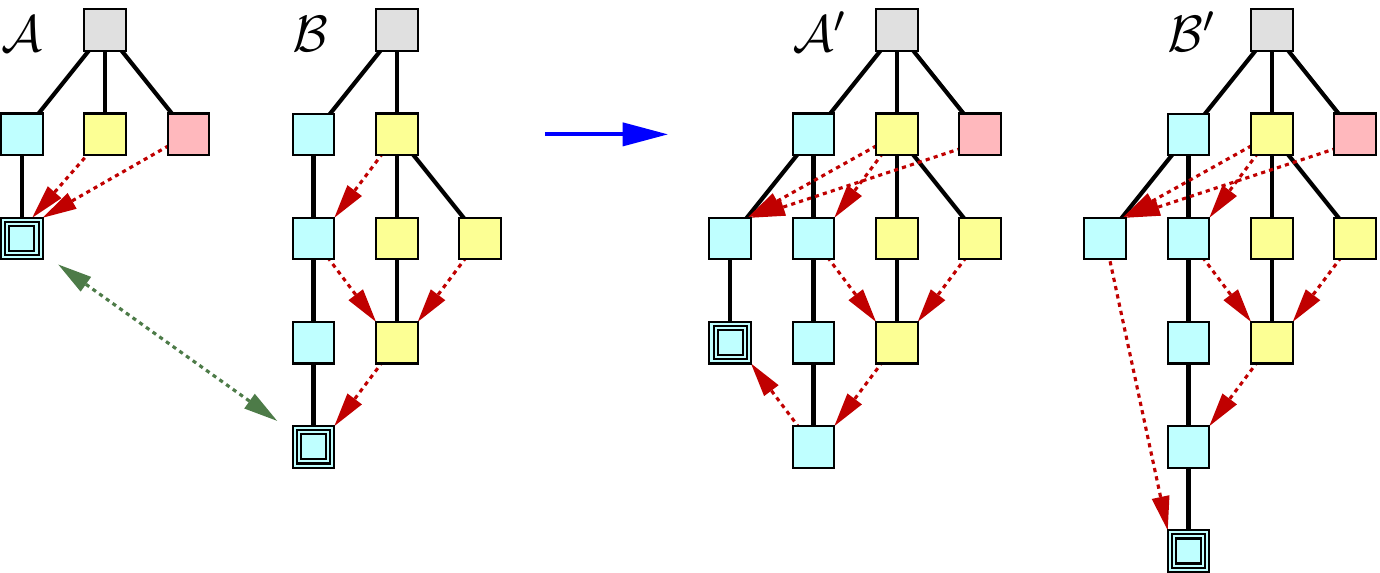}
\caption{The vistas $\mathcal A$ and $\mathcal B$ of two communicating agents. Since their heights are different, their updated versions $\mathcal A'$ and $\mathcal B'$ are generalized vistas, where red edges may go upward or skip multiple levels. The highlighted nodes are the bottom nodes.}
\label{fig:ssynch}
\end{figure}

\myskip
\mypar{Collective trees.} Since our algorithm allows agents to skip updating their states, it follows that they are no longer constructing vistas of the history tree of the network as defined in \cref{s:2}. Nonetheless, it is useful to define a structure that may act as a global reference incorporating all the (generalized) vistas that the agents are constructing.

The \emph{collective tree} at round~$t$, denoted by $\mathcal H_t$, is the structure obtained by match-and-merging together all the vistas constructed by the agents in a network up to the end of round~$t$. Of course, $\mathcal H_t$ depends on the local algorithm being executed by the agents, which affects the way they construct their vistas. For example, under the standard stabilizing algorithm outlined in \cref{s:2}, the collective tree $\mathcal H_t$ is simply the history tree of the network truncated at level $L_t$. Under the finite-state algorithm that we are going to describe in this section, $\mathcal H_t$ is constructed from generalized vistas, and may have little in common with the history tree.

We say that the collective tree \emph{acquires} a node $v$ at round~$t$ if $v$ is in $\mathcal H_{t}$ but not in $\mathcal H_{t-1}$. Similarly, the vista of an agent $p$ \emph{acquires} a node $v$ at round~$t$ if the vista of $p$ contains $v$ at round~$t$ but not at round~$t-1$.

The \emph{anonymity} $\mathbf a(v)$, where $v$ is a node in the collective tree, is defined as the total number of agents whose vista at some round has $v$ as the bottom node.

Observe that in generalized vistas and collective trees the notion of level is not well defined; we will in part restore this concept later through the definition of counting intervals.

\myskip
\mypar{Exposed pairs.} Another feature of our finite-state algorithm is the following. Recall that even if an agent updates its state in a given round, it may still choose to selectively discard some incoming messages from certain neighbors. However, there is a caveat: the algorithm guarantees that, if two agents send each other messages during a round, either both agents discard each other's messages or neither does. Thus, we make the following assumption.

\begin{assumption}\label{ass:1}
If an agent $p$ sends $k$ (identical) messages to another agent $q$ at round~$t$, then $q$ also sends $k$ (identical) messages to $p$ at the same round~$t$. Moreover, either all of these $2k$ messages are discarded, or none is.
\end{assumption}

Therefore, in any collective tree, whenever a node $v$ has an outgoing red edge of multiplicity $m_1>0$ to a child of a node $u$, then symmetrically $u$ has an outgoing red edge of multiplicity $m_2>0$ to a child of $v$. Moreover, the collective tree acquires a child of $v$ and a child of $u$ in the same round, namely, the first round in which an agent represented by $v$ exchanges messages with an agent represented by $u$.

With the above notation, if $v$ and $u$ are distinct nodes and both have a unique child in the collective tree, they are said to be an \emph{exposed pair}.

\begin{observation}\label{o:exposed}
If $v$ and $u$ are an exposed pair in $\mathcal H_t$, then the collective tree acquires the unique child of $v$ and the unique child of $u$ at the same round~$t'\leq t$.
\end{observation}

Recall from \cref{s:2} that the exposed pairs within a counting level are used to compute the Input Frequency function thanks to the equation:
\begin{equation}\label{eq:1}
m_1{\bf a}(u)=m_2{\bf a}(v)
\end{equation}
The same equation also holds in collective trees.

\begin{lemma}\label{l:exposed}
\cref{eq:1} holds for any exposed pair in a collective tree.
\end{lemma}
\begin{proof}
The left-hand side of \cref{eq:1} counts the total number of messages received by agents whose vista's bottom node is $u$ from agents whose vista's bottom node is $v$. Although these events may occur in different rounds, within the history of each agent represented by $u$ this happens exactly once: indeed, when it does, the agent updates its vista, and its bottom node is no longer $u$. Since each of the ${\bf a}(u)$ agents represented by $u$ receives $m_1$ messages from agents represented by $v$, the total number is $m_1{\bf a}(u)$.

Similarly, the right-hand side of \cref{eq:1} counts the total number of messages received by agents whose vista's bottom node is $v$ from agents whose vista's bottom node is $u$. The two sides of \cref{eq:1} are equal due to \cref{ass:1}.
\end{proof}

Note that the previous definition of exposed pair pertains to collective trees. We may give the same definition for vistas: two distinct nodes $v$ and $u$ in a vista $\mathcal V$ are an \emph{exposed pair} if each of them has a unique child in $\mathcal V$, and there is a red edge from $v$ to the child of $u$, and a red edge from $u$ to the child of $v$. However, we remark that, if $v$ and $u$ are an exposed pair in a vista, they are not necessarily an exposed pair in the collective tree, because they may have additional children in the collective tree that are missing from the vista. For this reason, we have the following:

\begin{observation}\label{o:notvista1}
\cref{o:exposed,l:exposed} may not hold for exposed pairs in a vista.
\end{observation}

\mypar{Counting intervals.} We now define an important structure that serves the same role as the counting levels of traditional history trees, but extends to generalized vistas and collective trees of disconnected networks. Recall from \cref{s:2} that a counting level is a level of a history tree where all nodes have a unique child. A counting level's outgoing red edges describe all interactions that occur in the corresponding round, and determine several exposed pairs in that level. If the network is connected at every round, this information can be used to compute the Input Frequency function via \cref{eq:1}.

Let $\mathcal H$ be a collective tree (or a generalized vista). A \emph{branch} in $\mathcal H$ is any path (using black edges only) from the root to a leaf. We define a \emph{cut} in $\mathcal H$ as a set of nodes $C$ such that any branch in $\mathcal H$ contains exactly one node of $C$.

We will now give a recursive definition of \emph{counting interval} (see \cref{fig:cut}, left). In a collective tree (or in a generalized vista) $\mathcal H$, a counting interval is a set of nodes $\mathcal I$ with the following properties:
\begin{itemize}
\item $\mathcal I$ is the union of $k+1\geq 2$ disjoint cuts of equal size $s\geq 1$. These cuts are called \emph{levels} and denoted by $L_0$, $L_1$, $L_2$, \dots, $L_k$.
\item For every $0\leq i<k$, each node in $L_{i+1}$ is the unique child of a node in $L_i$. Thus, the nodes of $\mathcal I$ can be partitioned into $s$ \emph{strands}, each of which consists of a node in $L_0$ and its $k$ descendants in $\mathcal I$.
\item For every $0\leq i<k$, if $v\in L_i$ has an outgoing red edge ending in $w\in \mathcal I$, then $w\in L_{i+1}$. Moreover, if $u\in L_i$ is the parent of $w$ and $v\neq u$, then $v$ and $u$ form an exposed pair in $\mathcal H$.
\item Consider the undirected graph on the $s$ strands of $\mathcal I$, where there is an edge between two strands if they contain two nodes not in $L_k$ that form an exposed pair. This graph is connected.
\item For every $0\leq i<k$ and for every node $v\in L_i$, the sub-vista of $\mathcal H$ at $v$ does not contain any counting intervals.
\item No proper subset of $\mathcal I$ is a counting interval in $\mathcal H$.
\end{itemize}

\begin{figure}
\centering
\includegraphics[scale=0.5]{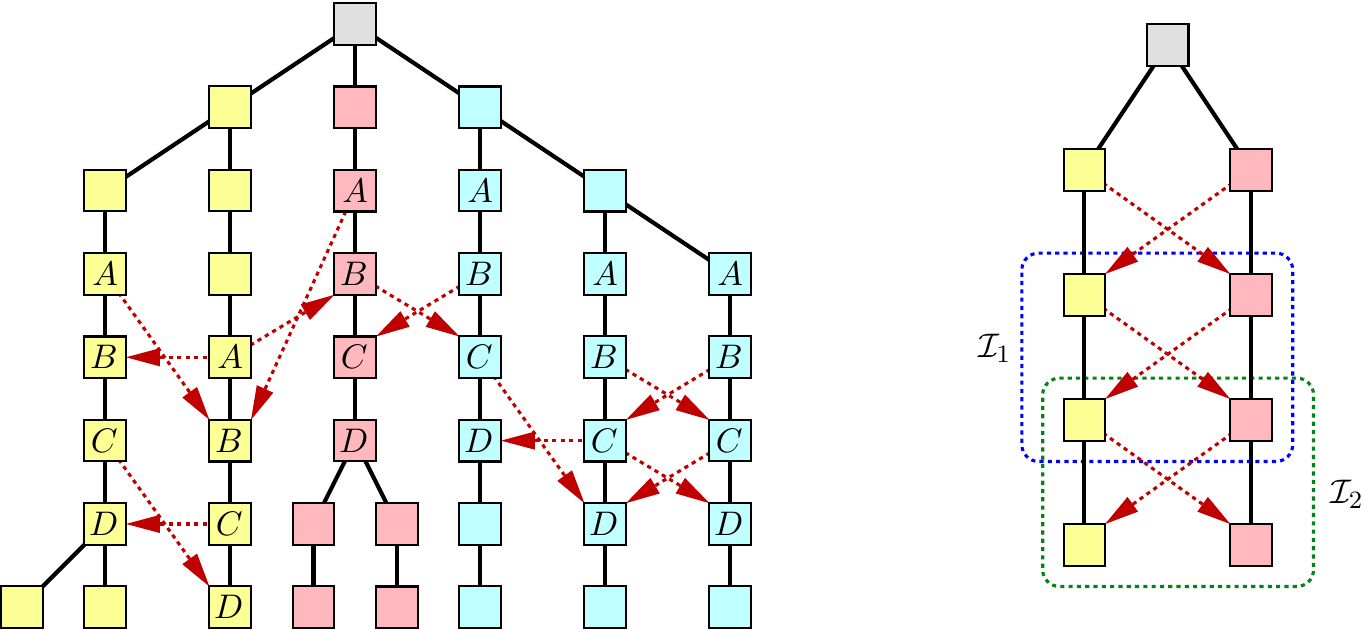}
\caption{Left: Sketch of a counting interval with four levels and six strands, where only the red edges involved in the relevant exposed pairs are displayed. Nodes with the same label are in the same level. Right: A collective tree with two counting intervals, where $\mathcal I_1$ dominates $\mathcal I_2$. Note that the sub-vista at either child of the root also contains a (trivial) counting interval, and therefore these two nodes cannot be in the first level of a counting interval.}
\label{fig:cut}
\end{figure}

Observe that, in a traditional history tree of a connected network, a counting interval is simply given by a counting level $L_i$ plus the subsequent level $L_{i+1}$, provided that none of the sub-vistas at the nodes of $L_i$ contains any counting level. In this sense, counting intervals may be regarded as a generalization of counting levels.

\myskip
\mypar{Dominance.} In a collective tree (or in a generalized vista), a counting interval $\mathcal I_1$ \emph{dominates} a counting interval $\mathcal I_2$ if every node of $\mathcal I_1$ has a strict descendant in $\mathcal I_2$.

\begin{lemma}\label{l:partial}
In any vista, dominance is a strict partial order on counting intervals.
\end{lemma}
\begin{proof}
The transitivity of the dominance relation follows directly from the transitivity of the descendant relation in a tree. As for the asymmetric property, consider a counting interval $\mathcal I_1$ that dominates a counting interval $\mathcal I_2$. By definition, any node $v$ in the last level of $\mathcal I_1$ has a descendant $u\in\mathcal I_2$. Since $v$ is the furthest node from the root within its strand of $\mathcal I_1$, all nodes in that strand are ancestors of $u$. All other nodes of $\mathcal I_1$ are unrelated to $u$, which proves that $\mathcal I_2$ does not dominate $\mathcal I_1$, and therefore dominance is asymmetric.

Since the dominance relation is transitive and asymmetric, it is a strict partial order.
\end{proof}

We stress that the definition of dominance does not contradict the definition of counting interval. That is, a counting interval $\mathcal I_1$ may dominate a counting interval $\mathcal I_2$ even though, as per the definition of counting interval, the sub-vista at any node of $\mathcal I_2$ (not in the last level) does not contain all nodes of $\mathcal I_1$. A simple example is found in \cref{fig:cut}, right.

For the next results, we need to introduce another feature of our finite-state algorithm. As we know, agents may not only discard each other's messages, but also skip updating their vistas altogether under certain conditions. One such condition is that the vista contains a counting interval.

\begin{assumption}\label{ass:2}
If an agent's vista at the beginning of round~$t$ does not contain a counting interval, the agent updates its vista during round~$t$.
\end{assumption}

We can now prove some important facts about counting intervals within collective trees.

\begin{lemma}\label{l:interval}
Let $\mathcal I$ be a counting interval in a collective tree, and let $L_0$, $L_1$, \dots, $L_k$ be the levels of $\mathcal I$. Then, there exists an integer $t$ such that, for all $0\leq i\leq k$, every node in $L_i$ is acquired by the collective tree at round~$t+i$.
\end{lemma}
\begin{proof}
By definition of counting interval, for every $0\leq i<k$ and for every node $v\in L_i$, the sub-vista $\mathcal V$ at $v$ does not contain any counting intervals. If $v$ is acquired by the collective tree at round~$r$, there is an agent $p$ whose vista is isomorphic to $\mathcal V$ at the end of round~$r$. Due to \cref{ass:2}, $p$ updates its vista at round~$r+1$, implying that the unique child of $v$ in $L_{i+1}$ is acquired at round~$r+1$.

It follows by induction that the desired property holds for each individual strand of $\mathcal I$. That is, the $k+1$ nodes within the same strand are acquired in $k+1$ consecutive rounds. Let $S_1$, $S_2$, \dots, $S_s$ be the strands of $\mathcal I$, and let $t_j$, with $1\leq j\leq s$, be the round at which the collective tree acquires the first node of $S_j$. Our goal is to prove that the $t_j$'s are all equal.

If two strands $S_j$ and $S_{j'}$ contain two nodes in $L_i$ with $0\leq i<k$ that form an exposed pair, their children in $L_{i+1}$ are acquired by the collective tree at the same round, due to \cref{o:exposed}. Thus, $t_j+i=t_{j'}+i$, and so $t_j=t_{j'}$. Now, the fact that all the $t_j$'s are equal immediately follows from the assumption that the exposed pairs in $\mathcal I\setminus L_k$ induce a connected graph on the strands, which holds by definition of counting interval.
\end{proof}

\begin{lemma}\label{l:total}
In any collective tree, dominance is a strict total order on counting intervals.
\end{lemma}
\begin{proof}
The transitive and asymmetric properties are proved as in~\cref{l:partial}. We only have to show that any two distinct counting intervals $\mathcal I_1$ and $\mathcal I_2$ in the same collective tree are comparable, i.e., one dominates the other.

By \cref{l:interval}, all nodes in the last level $L_{k_1}$ of $\mathcal I_1$ are acquired by the collective tree at the same round~$t_1$, and those in the last level $L_{k_2}$ of $\mathcal I_2$ are acquired at round~$t_2$. Also note that, given the way vistas are updated in our model of computation, if a node is acquired by the collective tree at round~$t$, its descendants are acquired strictly after round~$t$.

By the definition of counting interval, $L_{k_1}$ and $L_{k_2}$ are cuts of the collective tree. Hence, if $t_1<t_2$, every node in $L_{k_1}$ has a descendant in $L_{k_2}$, and so $\mathcal I_1$ dominates $\mathcal I_2$. Symmetrically, if $t_2<t_1$, then $\mathcal I_2$ dominates $\mathcal I_1$.

The only remaining case is $t_1=t_2$, which implies that $L_{k_1}=L_{k_2}$. If this is the case, then $\mathcal I_1\subseteq \mathcal I_2$ if $k_1<k_2$ and $\mathcal I_2\subseteq \mathcal I_1$ if $k_2<k_1$. This contradicts the definition of counting interval, which states that no proper subset of a counting interval is a counting interval. On the other hand, if $k_1=k_2$, then $\mathcal I_1=\mathcal I_2$, contradicting the assumption that the two counting intervals are distinct.

We conclude that either $\mathcal I_1$ dominates $\mathcal I_2$ or $\mathcal I_2$ dominates $\mathcal I_1$.
\end{proof}

Due to \cref{l:total}, if a collective tree has at least one counting interval, then there is a well-defined \emph{dominant} one, which dominates all others and whose strands are closest to the root.

We stress that, if an agent's vista at round~$t$ has a counting interval $\mathcal I$, the nodes of $\mathcal I$ as embedded in the collective tree $\mathcal H_t$ are not necessarily a counting interval, because a node with a unique child in the vista may have multiple children in $\mathcal H_t$ (recall that we made a similar observation about exposed pairs). Also, we have the following:

\begin{observation}\label{o:notvista2}
\cref{l:interval,l:total} may not hold for counting intervals in a vista. In particular, a vista containing counting intervals does not necessarily have a dominant one.
\end{observation}

\mypar{Algorithm overview.} We are finally ready to provide our finite-state algorithm. Its pseudocode is in \cref{l:3} of \cref{a:1}. Each agent sends its vista to all its neighbors and updates its own vista by match-and-merging incoming vistas as usual, with one exception: if the vista of an agent contains a dominant counting interval $\mathcal I$ (i.e., a counting interval that dominates all others within the vista), then all incoming vistas that have $\mathcal I$ as a dominant counting interval are discarded. Moreover, if all incoming vistas are discarded, the agent skips updating its own vista altogether for that round (hence it does not even add a child to the bottom node).

The rationale is that an agent with a dominant counting interval $\mathcal I$ ``believes'' that $\mathcal I$ is dominant for all agents and is sufficient to compute the Input Frequency function. Therefore, this agent deems it unnecessary to further update its vista, unless its belief is proven incorrect.

\myskip
\mypar{Algorithm details.} To make our algorithm more precise, we have to define the concept of \emph{isomorphism} between nodes, as well as counting intervals, in different vistas. Since we already have a graph-theoretic definition of isomorphism between vistas, an extension of the definition to counting intervals naturally follows. A node $v$ in a vista $\mathcal V$ is \emph{isomorphic} to a node $v'$ in a vista $\mathcal V'$ (or, equivalently, $v'$ is an \emph{isomorphic copy} of $v$) if the sub-vista at $v$ within $\mathcal V$ is isomorphic to the sub-vista at $v'$ within $\mathcal V'$. A counting interval $\mathcal I$ in a vista $\mathcal V$ is \emph{isomorphic} to a counting interval $\mathcal I'$ in a vista $\mathcal V'$ if there is a bijection $f\colon \mathcal I\to \mathcal I'$ such that every node $v\in \mathcal I$ is isomorphic to $f(v)\in \mathcal I'$. These definitions extend verbatim to collective trees.

According to our algorithm, if two communicating agents have dominant counting intervals in their respective vistas, and these counting intervals are isomorphic, both agents discard each other's messages. In all other cases (i.e., if either of them does not have a counting interval, or does not have a dominant counting interval, or both have dominant counting intervals which are not isomorphic), the agents use each other's messages to update their respective vistas. Hence \cref{ass:1} is indeed satisfied, and its direct consequences, such as \cref{o:exposed,l:exposed}, hold as well.

If, due to this rule, an agent ends up discarding all messages incoming from its neighbors (or if it has a dominant counting interval and received no messages), it does not update its vista at all, i.e., it does not even add a child to the bottom node of its vista. Thus, \cref{ass:2} also holds, and so do its consequences, such as \cref{l:interval,l:total}.

At the end of every round, if an agent's vista has a dominant counting interval $\mathcal I$, this is used to compute the Input Frequency function by repeated application of \cref{eq:1} on the exposed pairs within $\mathcal I$; the result is then returned as output by the agent. Otherwise, the agent simply returns the default output: its own input with a frequency of $100\%$.

\myskip
\mypar{Total agreement.} There are several challenges to proving the correctness of this relatively simple algorithm. Even if an agent has a dominant counting interval in its vista, using it to compute the Input Frequency function may lead to incorrect results, because \cref{eq:1} might not hold in that vista (cf.~\cref{o:notvista1}). Additionally, the rule that permits agents to discard messages could potentially result in a situation where all agents have an incorrect dominant counting interval, yet none of them updates its vista, preventing any progress from being made.

The next lemmas address these difficulties. We say that two agents \emph{agree} on a counting interval $\mathcal I$ at round~$t$ if both of their vistas at the end of round~$t$ contain isomorphic copies of $\mathcal I$, and $\mathcal I$ is actually a counting interval in both vistas. There is \emph{total agreement} on $\mathcal I$ if all agents in the system agree on $\mathcal I$.

\begin{lemma}\label{l:basicbranch}
Let $\mathcal{V}$ be a vista (or a collective tree) with a counting interval $\mathcal{I}$, and let $\mathcal{V}'$ be a vista (or a collective tree) that contains an isomorphic copy of $\mathcal{I}$. If $\mathcal{I}$, as embedded in $\mathcal{V}'$, is not a counting interval, then $\mathcal V'$ contains a branch that does not intersect the last level of $\mathcal I$ (hence the last level of $\mathcal{I}$ is not a cut of $\mathcal{V}'$).
\end{lemma}
\begin{proof}
By assumption, the sub-vistas of $\mathcal{V}'$ at the nodes in (the isomorphic copy of) $\mathcal{I}$ are isomorphic to the sub-vistas at the corresponding nodes of $\mathcal{V}$. Therefore, since $\mathcal{I}$ is a counting interval in $\mathcal{V}$, most defining properties of counting intervals also hold for $\mathcal{I}$ as embedded in $\mathcal{V}'$.

The only situation in which $\mathcal{I}$ is not a counting interval in $\mathcal{V}'$ is when $\mathcal{V}'$ contains a branch that includes nodes that are not present in $\mathcal{V}$ (and therefore are not in a sub-vista at any node of $\mathcal{I}$), which prevents some level of $\mathcal{I}$ from being a cut of $\mathcal{V}'$. In that case, the last level of $\mathcal I$ is also not a cut of $\mathcal V'$.

This includes, in particular, the case where a node in some level of $\mathcal I$ (other than the last) has multiple children in $\mathcal{V}'$, with only one child belonging to $\mathcal{I}$. Again, such a configuration implies the existence of a branch that does not intersect the last level of $\mathcal{I}$.
\end{proof}

\begin{corollary}\label{c:inherit}
If $\mathcal I$ is a counting interval in $\mathcal H_t$, then $\mathcal I$ is also a counting interval in any vista that contains an isomorphic copy of it at round~$t$.
\end{corollary}
\begin{proof}
By \cref{l:basicbranch}, if a vista $\mathcal V$ at round~$t$ contains an isomorphic copy of $\mathcal I$, but $\mathcal I$ is not a counting interval in $\mathcal V$, then there is a branch in $\mathcal V$ that does not intersect the last level $L_k$ of $\mathcal I$. Such a branch is also in $\mathcal H_t$ by definition of collective tree, which contradicts the fact that $L_k$ is a cut in $\mathcal H_t$.
\end{proof}

\begin{lemma}\label{l:collective}
If there is total agreement on a counting interval $\mathcal I$ at round~$t$, then $\mathcal I$ is also a counting interval in the collective tree $\mathcal H_t$. Moreover, the nodes in the last level of $\mathcal I$ represent a partition of all agents in the system.
\end{lemma}
\begin{proof}
Since all nodes of $\mathcal I$ are in all vistas at round~$t$, isomorphic copies of them are also found in $\mathcal H_t$. Thus, $\mathcal H_t$ contains an isomorphic copy of $\mathcal I$. Due to \cref{l:basicbranch}, if $\mathcal I$ is not a counting interval in $\mathcal H_t$, then there is a branch $B$ in $\mathcal H_t$ that does not intersect the last level of $\mathcal I$. By definition of collective tree, there is an agent whose vista $\mathcal V$ at round~$t$ contains the last node $b$ of $B$. Thus, $\mathcal V$ contains the sub-vista at $b$, which includes the entire branch $B$. It follows that the last level of $\mathcal I$ does not intersect $B$ in $\mathcal V$, and therefore $\mathcal I$ is not a counting interval in $\mathcal V$, contradicting our assumption. We conclude that $\mathcal I$ is a counting interval in $\mathcal H_t$.

Consider an agent $p$ and its vista $\mathcal V$ at round~$t$, and let $v$ be the bottom node of $\mathcal V$. Since $v$ is the unique sink in $\mathcal V$, it cannot be an ancestor of any node in $\mathcal I$. In particular, $v$ is either in the last level $L_k$ of $\mathcal I$ or a descendant of a unique node in $L_k$, because $L_k$ is a cut in $\mathcal V$. Hence, there is a unique node in $L_k$ that represents $p$. As this is true of every agent, the nodes in $L_k$ represent a partition of all the agents in the system.
\end{proof}

\begin{corollary}\label{c:collective}
If at round~$t$ there is total agreement on a counting interval $\mathcal I$ that is dominant in the vistas of all agents, then all agents return the correct output at round~$t$.
\end{corollary}
\begin{proof}
By \cref{l:collective}, $\mathcal I$ is a counting interval in $\mathcal H_t$, and the nodes in its last level $L_k$ represent a partition of all agents in the system. Since each level of $\mathcal I$ is a cut in $\mathcal H_t$, any two nodes in the same strand of $\mathcal I$ represent the same set of agents, and therefore have the same anonymity. The strands of $\mathcal I$ thus represent a partition of all the agents in the system; in order to compute the Input Frequency function, it suffices to determine the ratios between the anonymities of nodes in each pair of strands of $\mathcal I$.

If $\mathcal I$ is dominant in the vistas of all agents at round~$t$, then all agents use $\mathcal I$ to compute the Input Frequency function by repeated application of \cref{eq:1} on the exposed pairs within $\mathcal I\setminus L_k$. Since $\mathcal I$ is in $\mathcal H_t$, every agent correctly computes the ratio between the anonymities of the two nodes in each of these exposed pairs, due to \cref{l:exposed}. By definition of counting interval, the exposed pairs in $\mathcal I\setminus L_k$ induce a connected graph on the strands, which implies that every agent can compute the ratio between the anonymities of the nodes in any pair of strands of $\mathcal I$. We conclude that all agents correctly compute the Input Frequency function at round~$t$.
\end{proof}

\cref{c:collective} provides a sufficient condition for the algorithm to be correct. Thus, we only have to prove that eventually total agreement on a dominant counting interval is achieved and maintained.

\myskip
\mypar{Collective tree dynamics.} To begin with, the collective tree does not have any counting intervals, except in the trivial case where all agents have the same input (then the root of the collective tree and its unique child constitute a counting interval). We also remark that the existence of a counting interval in the collective tree does not necessarily imply the presence of a counting interval in any individual agent's vista. Next, we will study the way counting intervals are formed in the collective tree and how they may eventually become counting intervals in the agents' vistas.

Recall that a branch in the collective tree is any path from its root to a leaf. Thus, the total number of branches is equal to the number of leaves. We say that $k$ \emph{branchings} occur in the collective tree during round~$t$ if $\mathcal H_t$ has exactly $k$ more branches than $\mathcal H_{t-1}$. The collective tree \emph{acquires} a counting interval $\mathcal I$ at round~$t$ if $\mathcal I$ is a counting interval in $\mathcal H_{t}$ but not in $\mathcal H_{t-1}$. Similarly, the collective tree \emph{loses} a counting interval $\mathcal I$ at round~$t$ if $\mathcal I$ is a counting interval in $\mathcal H_{t-1}$ but not in $\mathcal H_{t}$.

The above definitions straightforwardly extend to vistas.

Since each agent starts with a vista containing a single branch, and the leaves of a vista represent disjoint and non-empty sets of agents, it follows that at most $n-1$ branchings may occur in each agent's vista. The same holds for the collective tree, for similar reasons. Thus, we make this important observation:

\begin{observation}\label{o:branching}
In total, at most $n-1$ branchings may occur in the collective tree, and at most $n(n-1)$ branchings may occur in the vistas of the $n$ agents.
\end{observation}

\begin{lemma}\label{l:never}
If there is total agreement on a counting interval $\mathcal I$ at round~$t$, then $\mathcal I$ is not lost by any of the vistas nor by the collective tree at any round~$t'>t$.
\end{lemma}
\begin{proof}
By \cref{l:collective}, $\mathcal I$ is a counting interval of $\mathcal H_t$ and the nodes in the last level $L_k$ of $\mathcal I$ represent all agents in the system. By \cref{l:basicbranch}, if the vista of an agent $p$ loses $\mathcal I$ at round~$t+1$, it acquires a branch that does not intersect $L_k$. The agents represented by the leaf within this branch are not represented by any node of $L_k$, which yields a contradiction.

Thus, none of the vistas loses $\mathcal I$ at round~$t+1$. In other words, there is total agreement on $\mathcal I$ at round~$t+1$ and, by \cref{l:collective}, $\mathcal I$ is a counting interval of $\mathcal H_{t+1}$, as well. Since we have proved our claim for $t'=t+1$, we can now use induction to generalize it to any $t'>t$.
\end{proof}

\begin{lemma}\label{l:branch}
If the collective tree (respectively, an agent's vista) loses a counting interval at round~$t$, then a branching occurs in the collective tree (respectively, in the same agent's vista) during round~$t$.
\end{lemma}
\begin{proof}
We will only prove our claim for collective trees, as the proof for vistas is identical. If $\mathcal H_{t-1}$ contains a counting interval $\mathcal I$, then clearly the nodes of $\mathcal I$ are also contained in $\mathcal H_{t}$. If $\mathcal I$ is not a counting interval of $\mathcal H_{t}$, then there is a branch in $\mathcal H_{t}$ that does not intersect the last level $L_k$ of $\mathcal I$, due to \cref{l:basicbranch}. Since $L_k$ is a cut in $\mathcal H_{t-1}$, this means that a branching occurs in the collective tree during round~$t$.
\end{proof}

\begin{lemma}\label{l:creation}
If for $\tau$ consecutive rounds there is no agent whose vista has a counting interval, and no branching occurs in the collective tree, then the collective tree acquires a counting interval during one of these $\tau$ rounds and never loses it.
\end{lemma}
\begin{proof}
Let $r$, $r+1$, \dots, $r+\tau-1$ be such rounds. By definition of the dynamic disconnectivity $\tau$, there exists a minimal integer $k\leq \tau$, as well as an integer $t$, with $r\leq t\leq r+\tau-k$, such that the union of the $k$ network multigraphs $G_t$, $G_{t+1}$, \dots, $G_{t+k-1}$ is connected. For all $0\leq i\leq k$, we define the level $L_i$ to be the set of leaves of the collective tree $\mathcal H_{t+i-1}$. We will prove that the $L_i$'s satisfy all the properties required by the definition of counting interval.

According to our algorithm, since no agents' vistas have counting intervals, no messages are discarded and all agents update their vistas in each of the $k$ consecutive rounds~$t$, $t+1$, \dots, $t+k-1$. Thus, if an agent $p$'s vista at the beginning of round~$t+i$, with $0\leq i<k$, has a bottom node isomorphic to a node $v\in L_i$, the vista of $p$ is updated at the same round, and its new bottom node is isomorphic to a node $v'\in L_{i+1}$. Moreover, since $v$ is a leaf of $\mathcal H_{t+i-1}$, such a $p$ exists, and therefore $v'$ is acquired by the collective tree at round~$t+i$.

We conclude that each leaf of the collective tree acquires a child at round $t+i$, for all $0\leq i<k$. Since no branching occurs in the collective tree during round~$t+i$, this child is unique, and hence all levels $L_i$ have the same size $s$, for all $0\leq i\leq k$. Thus, the $s$ strands are well defined. Also, each $L_i$ is a cut of $\mathcal H_{t+i-1}$, because it is the set of its leaves, and hence a cut in $\mathcal H_{t+k-1}$, as well.

Let us now turn to the configuration of the red edges. We argue that, for all $0\leq i<k$,  if two agents $p$ and $q$ send each other messages during round~$t+i$, and $p$'s bottom node $v$ is in $L_i$, then $q$'s bottom node $u$ is in $L_i$ as well. For otherwise, $u$ would not be a leaf in the collective tree, and the reception of a message from $p$ would cause $q$ to create a new child of $u$, contradicting the assumption that no branching occurs during round~$t+i$. This is because $v$ is acquired by the collective tree at round~$t+i-1$, and so none of the agents represented by $u$ could have received a vista with bottom node $v$ prior to round~$t+i$.

In the above situation, $v$ and $u$ become an exposed pair within $\mathcal H_{t+i}$ (provided that $v\neq u$), and hence within $\mathcal H_{t+k-1}$, as well. Also, since the union of the $k$ communication multigraphs from round~$t$ to round~$t+k-1$ is connected, it follows that all agents' bottom nodes are leaves at the beginning of each of these rounds. For otherwise, an agent with a leaf bottom node would eventually communicate with an agent whose bottom node is not a leaf, and a branching would occur as argued above. (Note that an agent of one type cannot become an agent of the other type.)

Therefore, during these $k$ rounds, a set of exposed pairs is created that induces a connected graph on the strands. In addition, by the minimality of $k$, the same is not true for any proper subset of these $k+1$ levels.

Finally, the fact that no sub-vista of a node in $L_i$ contains a counting interval is satisfied by assumption, for all $0\leq i<k$. We conclude that the collective tree acquires the counting interval $L_0\cup L_1\cup\dots\cup L_k$ at round~$t+k-1$.

Assume for a contradiction that the collective tree loses this counting interval at a subsequent round~$t'\geq t+k$. By \cref{l:basicbranch}, $\mathcal H_{t'}$ contains a branch that does not intersect $L_k$. However, this means that an ancestor of a node in $L_k$ has acquired a new child at round~$t'$, which contradicts the fact that all agents' bottom nodes were in $L_k$ at round~$t+k-1<t'$. Therefore, the counting interval is never lost.
\end{proof}

We remark that \cref{l:creation} only holds for collective trees and not for vistas.

\myskip
\mypar{Agent interactions.} Suppose that two agents $p$ and $q$ send each other messages at round~$t$, which are not discarded. Let $\mathcal I$ be a counting interval in the vista $\mathcal V$ of $p$ at the end of round~$t-1$, and assume that the vista $\mathcal W$ of $q$ at the end of round~$t-1$ does not contain $\mathcal I$. At the end of round~$t$, there are four possibilities:
\begin{itemize}
\item The vista of $p$ loses $\mathcal I$ and the vista of $q$ acquires $\mathcal I$. In this case, a branching occurs in the vista of $p$, due to \cref{l:branch}.
\item The vista of $p$ loses $\mathcal I$ and the vista of $q$ does not acquire $\mathcal I$. Also in this case, a branching occurs in the vista of $p$, due to \cref{l:branch}.
\item The vista of $p$ does not lose $\mathcal I$ and the vista of $q$ acquires $\mathcal I$.
\item The vista of $p$ does not lose $\mathcal I$ and the vista of $q$ does not acquire $\mathcal I$.
\end{itemize}

We will prove that, in the latter case, a branching occurs in the vista of $q$.

\begin{lemma}\label{l:interaction}
With the above notation, if the vista of $p$ does not lose $\mathcal I$ and the vista of $q$ does not acquire $\mathcal I$, then a branching occurs in the vista of $q$ at round~$t$.
\end{lemma}
\begin{proof}
By assumption, $\mathcal I$ is a counting interval in the vista $\mathcal V'$ of $p$ at round~$t$, but not in the vista $\mathcal W'$ of $q$ at round~$t$. However, $\mathcal W'$ contains an isomorphic copy of $\mathcal I$, because $\mathcal W'$ is obtained by match-and-merging $\mathcal V$ and $\mathcal W$ (plus, possibly, other vistas), and $\mathcal V$ contains $\mathcal I$.

Due to \cref{l:basicbranch}, $\mathcal W'$ contains a branch that does not intersect the last level $L_k$ of $\mathcal I$. This branch is not present in $\mathcal V'$, and therefore not in $\mathcal W$, because $\mathcal V'$ is obtained by match-and-merging $\mathcal V$ and $\mathcal W$ (plus, possibly, other vistas). We conclude that a branching occurs in $q$'s vista at round~$t$.
\end{proof}

\mypar{Undominated counting intervals.} We say that a counting interval $\mathcal I$ in an agent's vista $\mathcal V$ at round~$t$ is \emph{undominated} if no counting interval in $\mathcal V$ dominates $\mathcal I$. In other words, $\mathcal I$ is a maximal element in the dominance partial order within $\mathcal V$ (cf.~\cref{l:partial}). We stress that a vista may have multiple undominated counting intervals, none of which is dominant (cf.~\cref{o:notvista2}).

\begin{lemma}\label{l:undominated1}
If the counting interval $\mathcal I$ is undominated in a vista (or collective tree), then $\mathcal I$ is undominated in any vista (or collective tree) that contains $\mathcal I$ as a counting interval.
\end{lemma}
\begin{proof}
We will carry out the proof for vistas; the case of collective trees is handled identically. Let the vista $\mathcal V$ contain the undominated counting interval $\mathcal I$. For a contradiction, assume that a vista $\mathcal V'$ contains the counting intervals $\mathcal I$ and $\mathcal I'$, where $\mathcal I'$ dominates $\mathcal I$. Then, all nodes of $\mathcal I'$ are ancestors of nodes of $\mathcal I$, and hence are contained in every vista that contains $\mathcal I$. We infer that $\mathcal V$ also contains $\mathcal I'$, which is not a counting interval in $\mathcal V$, or else it would dominate $\mathcal I$.

By \cref{l:basicbranch}, $\mathcal V$ contains a branch $B$ that does not intersect the last level of $\mathcal I'$. However, this implies that $B$ does not intersect the last level of $\mathcal I$ either, contradicting the fact that this level is a cut of $\mathcal V$.
\end{proof}

\begin{lemma}\label{l:undominated2}
Let $p$ be an agent whose vista at the beginning of round~$t$ contains an undominated counting interval $\mathcal I$, and let $q$ be an agent whose vista at the beginning of round~$t$ does not have $\mathcal I$ as a counting interval. Then, if $p$ and $q$ are neighbors at round~$t$, they do not discard each other's messages.
\end{lemma}
\begin{proof}
Our algorithm makes $p$ and $q$ discard each other's messages only if they have the same dominant counting interval at the beginning of round~$t$. If this is the case, since $\mathcal I$ is undominated in the vista of $p$, then $\mathcal I$ is dominant. However, by assumption the vista of $q$ does not have $\mathcal I$ as a counting interval, and therefore it does not have the same dominant counting interval as the vista of $p$.  Hence, $p$ and $q$ do not discard each other's messages.
\end{proof}

\mypar{Classification of rounds.} We classify the rounds into three types:
\begin{enumerate}
\item[(i)]Rounds at the beginning of which no agent's vista has a counting interval.  
\item[(ii)]Rounds at the beginning of which some agents' vistas have a counting interval, but there is no total agreement on any counting interval.
\item[(iii)]Rounds at the beginning of which there is total agreement on a counting interval.
\end{enumerate}

Since the network is $\tau$-union-connected, it is convenient to define a \emph{block} of rounds as any sequence of $\tau$ consecutive rounds. Thus, a block of type~(i) is a block in which all $\tau$ rounds are of type~(i), etc.

\begin{lemma}\label{l:final1}
There can be at most $2n-1$ disjoint (but not necessarily contiguous) blocks of type~(i) before the first round of type~(iii) occurs.
\end{lemma}
\begin{proof}
During rounds of type~(i), according to our algorithm, no messages are discarded. If no branching occurs in the collective tree for an entire block of type~(i), a new counting interval $\mathcal I$ is acquired by the collective tree and is never lost, due to \cref{l:creation}. Recall from \cref{o:branching} that a branching in the collective tree may occur at most $n-1$ times, and hence such a counting interval $\mathcal I$ is acquired within $n$ blocks of type~(i).

After that, since $\mathcal I$ is never lost by the collective tree, in every block of type~(i) progress is made toward total agreement on $\mathcal I$ by the agents. Specifically, for every node $v\in \mathcal I$, within each block of type~(i) at least one more agent's vista acquires $v$ (by definition of $\tau$, and since no messages are discarded), until all agents have acquired $v$. Thus, it takes at most $n-1$ disjoint blocks of type~(i) for all agents to acquire all nodes of $\mathcal I$. At this point total agreement on $\mathcal I$ is achieved, because \cref{c:inherit} implies that $\mathcal I$ is a counting interval of any vista that contains it.

In conclusion, there can be at most $2n-1$ disjoint (not necessarily contiguous) blocks of type~(i) before total agreement on a counting interval is reached and a round of type~(iii) occurs.
\end{proof}

\begin{lemma}\label{l:final2}
There can be at most $n+2b_1-1$ disjoint (but not necessarily contiguous) blocks of type~(ii), where $b_1$ is the total number of branchings that occur in agents' vistas during these blocks, before the first round of type~(iii) occurs.
\end{lemma}
\begin{proof}
By definition, in a round of type~(ii) there are agents whose vistas contain counting intervals, but no total agreement on any of them. Let $\mathcal I$ be any undominated counting interval in an agent's vista. Due to \cref{l:branch}, in every vista that loses $\mathcal I$ a branching must occur. Moreover, in every block of type~(ii), at least one agent $p$ whose vista has $\mathcal I$ as a counting interval exchanges messages with an agent $q$ whose vista does not have $\mathcal I$ as a counting interval. Due to \cref{l:undominated2}, these messages are not discarded.

Hence, by our previous analysis of agent interactions and by \cref{l:interaction}, a branching occurs in the vista of $p$ or in the vista of $q$, or the counting interval $\mathcal I$ is acquired by the vista of $q$ and is not lost by the vista of $p$. By \cref{l:undominated1}, $\mathcal I$ is undominated in any vista that acquires it, and therefore the same argument applies to the next blocks of type~(ii).

Thus, in every block of type~(ii), either the number of agents whose vistas have the counting interval $\mathcal I$ increases, or it remains stable and at least one branching occurs in the vista of some agent, or it decreases by $k>0$ while at least $k$ branchings occur in the vistas of some agents.

Note that, if a round of type~(iii) does not occur, a block of type~(ii) is always followed by another block of type~(ii), unless all counting intervals are lost by all vistas. However, even if $\mathcal I$ is lost by all vistas, another undominated counting interval will replace it at the next round of type~(ii).

Track one undominated counting interval at a time, choosing a replacement whenever the current one is lost. In every completed phase, each increase is compensated by a later decrease; only the final phase contributes at most $n-1$ additional increases before total agreement is reached. Since decreases and stable blocks account for at most $b_1$ branchings, there are at most $(n-1)+b_1$ increasing blocks and at most $b_1$ non-increasing blocks. Thus, before a round of type~(iii) occurs, there are at most $(n-1)+2b_1$ disjoint (but not necessarily contiguous) blocks of type~(ii).
\end{proof}

\begin{lemma}\label{l:final3}
There are at most $2b_2$ consecutive blocks of type~(iii), where $b_2$ is the total number of branchings that occur in agents' vistas during rounds of type~(iii), before total agreement on a counting interval that is dominant in the vistas of all agents is achieved.
\end{lemma}
\begin{proof}
Consider the first round~$t$ in such a sequence of blocks of type~(iii). By definition, at this round there is total agreement on a counting interval $\mathcal I$, which is also a counting interval in the collective tree $\mathcal H_t$, due to \cref{l:collective}. Since $\mathcal H_t$ has at least one counting interval, it has a dominant counting interval $\mathcal D$, by \cref{l:total}. The nodes of $\mathcal D$ are ancestors of nodes of $\mathcal I$, and hence all agents' vistas contain all of them. Therefore, by \cref{c:inherit}, there is total agreement on the counting interval $\mathcal D$. By \cref{l:never}, $\mathcal D$ is never lost by any of the agents' vistas nor by the collective tree. Thus, by \cref{l:undominated1}, $\mathcal D$ is undominated in all vistas, as well as dominant in the collective tree, at all rounds~$t'\geq t$.

We will show that eventually $\mathcal D$ becomes the dominant counting interval in all vistas. Consider an agent's vista $\mathcal V$ where $\mathcal D$ is not dominant, albeit undominated. Then, there is another undominated counting interval $\mathcal D'$ in $\mathcal V$. Observe that there cannot be total agreement on $\mathcal D'$. For otherwise $\mathcal D'$ would be a counting interval in the collective tree (by \cref{l:collective}), and it would be undominated in the collective tree as well (by \cref{l:undominated1}), contradicting the fact that $\mathcal D$ is the dominant counting interval in the collective tree.

Since $\mathcal D$ is undominated in all vistas at all rounds~$t'\geq t$, no counting interval other than $\mathcal D$ can be dominant in any vista at any round~$t'\geq t$. Therefore, none of the messages between agents whose vistas contain $\mathcal D'$ and agents whose vistas do not contain $\mathcal D'$ are discarded (because such vistas cannot have the same dominant counting interval). Thus, as we argued in \cref{l:final2}, during each block either the number of agents whose vistas contain the counting interval $\mathcal D'$ increases, or it remains stable and at least one branching occurs in some vista, or it decreases by $k>0$ and at least $k$ branchings occur in the vistas of some agents.

Let $x$ be the number of agents' vistas containing $\mathcal D'$, which varies between $0$ and $n-1$. Every block in which $x$ remains stable accounts for at least one branching, and a decrease by $k>0$ accounts for at least $k$ branchings. After the last branching, if $x>0$, it would increase in every block until total agreement on $\mathcal D'$ was reached, a contradiction. Hence $x$ eventually reaches $0$, and its total increase cannot exceed its total decrease. Therefore, there are at most $b_2$ blocks in which $x$ increases and at most $b_2$ in which it does not, so $\mathcal D'$ is lost by all vistas within $2b_2$ consecutive blocks.

Whenever $\mathcal D'$ is lost while $\mathcal D$ is not yet dominant in all vistas, select another undominated counting interval distinct from $\mathcal D$ and repeat the argument. Tracking one such interval at a time partitions the blocks into disjoint phases, so the same accounting gives at most $b_2$ increasing and $b_2$ non-increasing blocks in total. Hence, within $2b_2$ consecutive blocks of type~(iii), no undominated counting interval other than $\mathcal D$ remains in any vista, and $\mathcal D$ becomes dominant in all vistas.
\end{proof}

\myskip
\mypar{Algorithm correctness.} We conclude this section with a proof of correctness of our finite-state algorithm.

\begin{theorem}\label{t:3}
There is a finite-state universal algorithm that operates in any $\tau$-union-connected anonymous dynamic network of unknown size $n$ (and unknown $\tau$) and stabilizes in at most $\tau(2n^2+n)$ rounds.
\end{theorem}
\begin{proof}
Observe that, as soon as a round of type~(iii) occurs, all subsequent rounds will be of type~(iii). Indeed, as soon as there is total agreement on a counting interval, by \cref{l:never} there will always be total agreement on the same counting interval. Therefore, any execution of our algorithm consists of an alternation between rounds of type~(i) and type~(ii), followed by an infinite sequence of rounds of type~(iii).

By \cref{l:final1}, there are at most $2n-1$ disjoint blocks of type~(i), and by \cref{l:final2} there are at most $n+2b_1-1$ disjoint blocks of type~(ii). The remaining rounds of type~(i) can be grouped into maximal contiguous sequences that are shorter than $\tau$ rounds, and the same goes for the rounds of type~(ii). We call these \emph{spare sequences}.

Note that a spare sequence of type~(ii) that is followed by a round of type~(i) contains at least one round where an agent's vista loses a counting interval; thus, a branching occurs in that vista, due to \cref{l:branch}. If $b_3$ is the total number of branchings that occur in agents' vistas during spare sequences of type~(ii), we may have at most $2b_3+2$ spare sequences in total.

Finally, by \cref{l:final3}, after at most $2b_2$ consecutive blocks of type~(iii), total agreement on a dominant counting interval is achieved. By \cref{c:collective}, at this point all agents return the correct output. Moreover, since all agents now have the same dominant counting interval, they stop updating their states and keep returning the same correct output.

Observe that $b_1$, $b_2$ and $b_3$ count branchings occurring in agents' vistas during three pairwise disjoint sets of rounds. Hence, by \cref{o:branching}, we have $b_1+b_2+b_3\leq n(n-1)$.

We conclude that the algorithm is correct, and stabilization is achieved within
\[
\begin{aligned}
&(2n-1)+(n+2b_1-1)+(2b_2)+(2b_3+2)\\
&\qquad=3n+2(b_1+b_2+b_3)\leq 3n+2n(n-1)=2n^2+n
\end{aligned}
\]
consecutive blocks, and therefore the total stabilization time is at most $\tau(2n^2+n)$ rounds. Note that a generalized vista of height $O(\tau n^2)$ can be encoded in $O(\tau n^4\log n)$ bits, implying that the algorithm is finite-state.
\end{proof}

\section{Concluding Remarks}\label{s:6}
We have proposed the first self-stabilizing universal algorithm for anonymous dynamic networks; this algorithm has a linear stabilization time (\cref{s:4}). We have also provided the first finite-state universal algorithm for anonymous dynamic networks; this algorithm has a quadratic stabilization time (\cref{s:5}). Both algorithms also work in $\tau$-union-connected networks, at the cost of a factor of $\tau$ in their running times (which is optimal, see~\cite{DVdisc}).

It is natural to ask whether a universal algorithm that is both self-stabilizing and finite-state exists. For static networks, such an algorithm is found in~\cite{BV02b}; as for dynamic networks, we gave an optimal solution in \cref{s:3} assuming that the number of agents is known.

Another open problem is to improve the stabilization times of our algorithms. For self-stabilizing algorithms in dynamic networks, we believe that a linear dependence on the garbage coefficient is unavoidable, and therefore our algorithm of \cref{s:4} is asymptotically optimal. As for finite-state algorithms, we believe that a linear stabilization time is achievable. Possibly, an improved analysis of our algorithm in \cref{s:5} may already yield a linear stabilization time.

The generalized theory of history trees developed in \cref{s:5} naturally lends itself to modeling \emph{semi-synchronous} networks, where some agents may unpredictably become inactive and skip a round of interactions and state updates. We leave a detailed study of this setting for future work.

Finally, it would be interesting to investigate universal algorithms that are not only finite-state, but also allow agents to stop sending messages after stabilization has been achieved.

\section*{Acknowledgements}
The authors would like to thank the anonymous reviewers of Theoretical Computer Science for suggestions that improved the readability of this paper.

\bibstyle{plainurl}
\bibliography{selfADN}

\newpage
\appendix

\section{Pseudocode}\label{a:1}
This section contains the pseudocode for all the algorithms described in \cref{s:3,s:4,s:5}. The entry point is the function \texttt{Main()}. Each agent runs an independent instance of this function at every round $t\geq1$ and has private instances of the local variables.

The only primitives are the following functions:
\begin{itemize}
\item \texttt{Input()} returns the agent's input.
\item \texttt{Output()} makes an agent return an output.
\item \texttt{SendToAllNeighbors()} takes a message as an argument and sends it to all agents that share a link with the caller in the current communication round.
\item \texttt{ReceiveFromAllNeighbors()} returns a multiset of messages coming from all incident links. These are the messages that have been passed to the function \texttt{SendToAllNeighbors()} by the neighboring agents in the current round.
\item \texttt{Chop()} takes a vista of a history tree and chops it, eliminating level $L_0$, as described in \cref{s:3}.
\item \texttt{ComputeFrequencies()} takes a vista of a history tree and applies the algorithm in \cite[Theorem~3.2]{DVdisc} on it, which returns either (an estimate of) the frequencies of all inputs or ``Unknown''.
\end{itemize}

\lstset{style=mystyle}

\begin{lstlisting}[caption={Universal self-stabilizing finite-state algorithm for networks of known size $n$ and known dynamic disconnectivity $\tau$ (\cref{s:3})\label{l:1}},captionpos=t,float=ht!,mathescape=true]
# Every agent has a private copy of this internal variable:
myHT        # current vista of the history tree

function Main($n$, $\tau$)   # the network parameters $n$ and $\tau$ are given
    if myHT is not a coherent vista with at most $\tau(2n-2)$ levels or is oversized, then
        myHT := vista with two nodes: a root with a child labeled Input()
    SendToAllNeighbors(myHT)
    receivedMessages := ReceiveFromAllNeighbors()
    allMessages := $\text{receivedMessages} \cup \{\text{myHT}\}$
    minHT := $\arg\min_{\text{HT}\in \text{allMessages}}$ HT.height
    while myHT.height > minHT.height, do Chop(myHT)
    add a child to myHT.bottom and label it Input()   # this is the new bottom
    for all HT in receivedMessages, do
        while HT.height > minHT.height, do Chop(HT)
        match-and-merge HT into myHT
        add a red edge from HT.bottom to myHT.bottom
    if myHT has $\tau(2n-2)+1$ levels, then Chop(myHT)
    F := ComputeFrequencies(myHT)
    if F is "Unknown", then Output($\{(\text{Input(), 100\%)}\}$)
    else Output(F)
\end{lstlisting}

\begin{lstlisting}[caption={Universal self-stabilizing algorithm for networks of unknown size and unknown dynamic disconnectivity (\cref{s:4})\label{l:2}},captionpos=t,float=ht!,mathescape=true]
# Every agent has a private copy of these internal variables:
myHT        # current vista of the history tree
myFlag      # controls the deletion of the first level of HT (either 0 or 1)

function Eval(HT, flag)
    return 2 * HT.height + flag

function Main()
    if myHT or myFlag is malformed, then
        myHT := vista with two nodes: a root with a child labeled Input()
        myFlag := 1
    myMessage := (myHT, myFlag)
    SendToAllNeighbors(myMessage)
    receivedMessages := ReceiveFromAllNeighbors()
    allMessages := $\text{receivedMessages} \cup \{\text{myMessage}\}$
    (minHT, minFlag) := $\arg\min_{\text{(HT, flag)}\in \text{allMessages}}$ Eval(HT, flag)
    while myHT.height > minHT.height, do Chop(myHT)
    add a child to myHT.bottom and label it Input()   # this is the new bottom
    for all (HT, flag) in receivedMessages, do
        while HT.height > minHT.height, do Chop(HT)
        match-and-merge HT into myHT
        add a red edge from HT.bottom to myHT.bottom
    myFlag := 1 - minFlag
    if myFlag = 1, then Chop(myHT)
    F := ComputeFrequencies(myHT)
    if F is "Unknown", then Output($\{(\text{Input(), 100\%)}\}$)
    else Output(F)
\end{lstlisting}

\begin{lstlisting}[caption={Universal finite-state algorithm for networks of unknown size and unknown dynamic disconnectivity (\cref{s:5})\label{l:3}},captionpos=t,float=ht!,mathescape=true]
# Every agent has a private copy of this internal variable:
myHT        # current vista of the history tree (initially a root with a child labeled Input())

function Main()
    SendToAllNeighbors(myHT)
    receivedMessages := ReceiveFromAllNeighbors()
    relevantVistas := empty multiset
    for all HT in receivedMessages, do
        if myHT and HT do not have isomorphic dominant counting intervals, then
            relevantVistas := $\text{relevantVistas} \cup \{\text{HT}\}$
    if relevantVistas is not empty or myHT has no dominant counting interval, then
        add a child to myHT.bottom and label it Input()   # this is the new bottom
        for all HT in relevantVistas, do
            match-and-merge HT into myHT
            add a red edge from HT.bottom to myHT.bottom
    if myHT has a dominant counting interval, then
        F := frequencies computed using the dominant counting interval of myHT
        Output(F)
    else Output($\{(\text{Input(), 100\%)}\}$)
\end{lstlisting}

\end{document}